\documentclass{article}

\usepackage{arxiv}

\usepackage[utf8]{inputenc} 
\usepackage[T1]{fontenc}    
\usepackage{hyperref}       
\usepackage{url}            
\usepackage{booktabs}       
\usepackage{amsfonts}       
\usepackage{nicefrac}       
\usepackage{microtype}      
\usepackage{graphicx}
\usepackage{natbib}
\usepackage{doi}
\usepackage{amsmath}
\usepackage{amsthm}
\usepackage{amssymb}
\usepackage{tikz}
\usepackage{xcolor}
\usepackage{pgfplots}
\pgfplotsset{compat=1.17,
	legend image code/.code={
		\draw[mark repeat=2,mark phase=2]
		plot coordinates {
			(0cm,0cm)
			(0.15cm,0cm)        
			(0.3cm,0cm)         
		};%
}}
\usepackage{adjustbox}
\usepackage[subrefformat=parens,labelformat=parens]{subfig}
\usepackage{comment}
\DeclareMathOperator*{\argmax}{arg\,max}

\usepackage[colorinlistoftodos,textsize=tiny]{todonotes}

\def\checkmark{\tikz\fill[scale=0.4](0,.35) -- (.25,0) -- (1,.7) -- (.25,.15) -- cycle;} \newcommand{\jospin}{\footnotesize{\color{magenta}{LJ}}}
\newcommand{\chirac}{\footnotesize{\color{blue}{JC}}}
\newcommand{\macron}{\footnotesize{\color{orange}{EM}}}
\newcommand{\bayrou}{\footnotesize{\color{cyan}{FB}}}
\newcommand{\meluche}{\footnotesize{\color{red}{JLM}}}
\newcommand{\green}{\footnotesize{\color{olive}{NM}}}
\newcommand{\hamon}{\footnotesize{\color{magenta}{BH}}}
\newcommand{\fillon}{\footnotesize{\color{blue}{FF}}}
\newcommand{\lepen}{\footnotesize{\color{black}{MLP}}}

\def\mavr{{\rm MAV}^R}
\def\ccav{{\rm CCAV}}
\def\ccavr{{\rm CCAV}^R}
\def\sccav{\mbox{\rm S-CCAV}}
\def\sccavr{\mbox{\rm S-CCAV}^R}

\newtheorem{definition}{Definition}

\newtheorem{theorem}{Theorem}
\newtheorem{proposition}{Proposition}

\newtheorem{example}{Example}
\newcommand*{\drule}[2]{\textbf{#1}: \emph{#2}}

\title{Approval with Runoff}
\author{
    Théo Delemazure \textsuperscript{\rm (1)}, 
    Jérôme Lang \textsuperscript{\rm (2)},
    Jean-François Laslier \textsuperscript{\rm (3)},
    Remzi Sanver\textsuperscript{\rm (2)}\\
    \textsuperscript{\rm (1)} LAMSADE, Université Paris Dauphine, PSL, CNRS\\
    \textsuperscript{\rm (2)} LAMSADE, CNRS, Université Paris Dauphine, PSL\\
    \textsuperscript{\rm (3)} CNRS, Paris School of Economics, PSL
}
\begin{document}

\maketitle
\begin{abstract}
We define a family of runoff rules that work as follows: voters cast approval ballots over candidates; two finalists are selected; and the winner is decided by majority. With approval-type ballots, there are various ways to select the finalists. We leverage known approval-based committee rules and study the obtained runoff rules from an axiomatic point of view. Then we analyze the outcome of these rules on single-peaked profiles, and on real data.
\end{abstract}

\section{Introduction}


Plurality with runoff (also known as runoff voting) is a widely used single-winner voting rule, in fact the most common rule for presidential elections throughout the world\footnote{See \url{https://en.wikipedia.org/wiki/Two-round_system}.}. But the social choice literature has pointed out that plurality with runoff suffers from so many drawbacks that we may wonder why it is used at all: it is highly sensitive to cloning, fails monotonicity, reinforcement, participation, Condorcet-consistency, and is very easy to manipulate. In particular, its high sensitivity to cloning has a number of derived effects before the vote (at the level of the determination of candidates) and at voting time (with massive strategic voting of a specific kind, named ``useful voting''). Perhaps the main reason why it is so widely used after all is related to the fact that runoff voting is not used as a one-shot voting rule but as two-round protocol: voters are called to urns for the first round, the results are made public,  and then some amount of time passes (typically one or two weeks). And in between the two rounds, many things happen. 

In most variants, only two candidates are selected for the runoff. The others candidates may negotiate their support to one of the two contenders, leading to adjustments in the platforms proposed in the second round. The TV debates that take place between the two finalists at that point in time are considered as the most important moment in the whole campaign, and many voters may, during this period, review their decision to participate or not to the second vote. For all these reasons, the existence of two rounds of vote separate in time is considered to be crucial for the voters' information.

Are the informational benefits of a runoff protocol enough to overcome its numerous theoretical drawbacks? There can be diverse opinions about this. However, instead of answering this question, we may ask another one: is it possible to keep the nice benefit of the two-round protocol without having to bear all the drawbacks of \textit{Plurality} at the first round? 

Clearly, if the answer to this question is positive, the format of the ballots at first round must no longer be uninominal. Several possibilities exist: ordinal ballots, cardinal ballots, or more simply, approval ballots. Approval ballots have several advantages; to start with, they are simple and easy to express.

In this paper we explore this possibility seriously.
We define an approval-with-runoff election as a two-round protocol:
\begin{enumerate}
    \item First round: voters cast approval ballots, from which the
    two finalists are selected.
    \item Second round: voters cast votes for one of the two finalists, and the majority winner wins the overall election.\footnote{
    The present paper is concerned with single-winner elections. Approval voting with a runoff is effectively used in several cantons in Switzerland for multi-winner elections. 
    The precise rules vary from one canton to the other so that the second round is sometimes almost unused, as in the canton of Zurich (\citep{LaslierVDS16}, \citep{VanderStraetenetal18}). }
\end{enumerate}

Formally, we define approval-with-runoff as a {\em voting rule}, with a one-shot input, and study its properties in a similar way as we would study the properties of plurality with runoff. 
Then two major questions arise:
\begin{enumerate}
    \item What should the input of the rule consist of?
    \item Which rule should be used to determine the two finalists?
\end{enumerate}

For question 1, the answer becomes clear once we remark that we need the approval data for computing the finalists, and the pairwise comparisons between candidates for computing the final winner. Of course, we will not need {\em all} comparisons between arbitrary pairs of candidates; but just as plurality with runoff, seen as a voting rule, takes full rankings as input although most of this information will not be asked, here too,  we need more information in the input than we will ask voters, and the normative properties of the rule will be evaluated with respect to this (mostly private) information. Now, requiring pairwise comparisons between all pairs of candidates just means that we need each voter's ranking of candidates, and requiring her approval set means that this ranking comes with a threshold that separates approved candidates from disapproved candidates. This data structure is called an {\em approval-preference (AP)} profile \citep{BramsS09}.

Notice that, with respect to the points mentioned in the introduction, the framework that we use does not allow taking into account the evolution of voters and candidates in between the two rounds. We leave these problems to further research and, in this paper as it is often the case in social choice theory, we concentrate on the counting of sincere ballots cast by a fixed electorate. 

For question 2, things are more complex because there is not a unique way to select two candidates from approval ballots. 
The general setting in which we select $k$ candidates (here, $k = 2$) from an approval profile is called an {\em approval-based committee rule} (ABC rule).
A recent and extensive survey is on ABC rules is \cite{LacknerS20}, and we have now series of results that tell us which properties these various rules satisfy and for which contexts they are suitable. 
Most importantly, the choice of the rule used for the first round has strong implications about the very nature of the two-stage rule, both from a normative point of view and from a political science point of view: should we send to the second round the most two approved candidates? Or should we offer the voters two candidates that are diverse enough? Should we pay attention to proportionality issues? Should we guarantee the most approved candidate is among the two finalists? 

Our primary aim is {\em to define approval with runoff} not just as one rule but {\em as a family of rules}, and {\em to explore the reasons that may guide us towards the choice of one of the rules in the family}. 


The paper is organized as follows.
We start by related work (Section \ref{related}). We define and study the family of Approval-based committee rules (Section \ref{abc-defin}) together with a selection of meaningful rules. Then, we define Approval with runoff rules and study these rules form an axiomatic point of view (Section \ref{avr-defin}). We analyse the outcome of these rules on one-dimensional Euclidean profiles (Section \ref{sp}), and move on to applying the rules on real data (Section \ref{experiments}). We conclude in
Section \ref{conclu}.

\section{Related work}\label{related}

Approval with Runoff was first introduced in \citep{remzi2010} and compared to other rules based on approval-preference profiles.
\citep{Green-ArmytageT20}
consider plurality with runoff together with eight other runoff rules for selecting the finalists, with varying input formats (ordinal, approval, numerical), including Approval with Runoff.
Voters are supposed to vote sincerely and, for Approval voting, to approve a candidate if and only if the utility they give to this candidate is larger than the average utility of all candidates running. They evaluate these rules along four numerical criteria (expected utility of winner, of the runoff loser, representativeness, resistance to strategy); numerical results come both from using real data and from simulations. Among other conclusions, plurality with runoff scores particularly bad, and approval with runoff, slightly better, although it is beaten by plurality with runoff on two criteria: representativeness and resistance to strategy. 

A runoff can also be seen as an extreme case of shortlisting (with at most two selected candidates). Using approval for shortlisting candidates was studied recently in \citep{LacknerM21}; a crucial difference with runoff rules is that shortlisting does not impose constraints on the number of selected candidates, which leads to very different rules (such as dichotomy rules, or rules based on large gaps).

Approval-based multiwinner election rules have received enormous attention these last ten years: see \citep{LacknerS20} for a review.
They are clustered in several groups according to the objective of the selection: excellence (select the individually best $k$ candidates), proportional representation (ensure that each coherent group of voters is represented in the selection, proportionally to its size), or diversity (output a diverse set of candidates, avoid similar candidates in the selection).
Which of these three clusters of rules suits the selection of  runoff candidates better is not clear at this point; our paper aims at answering (at least partly) this question. 

Defining voting rules that take as input approval-preference profiles has been initiated in \citep{BramsS09}, who propose and study two such rules (preference approval voting and fallback voting) that have been studied in a number of subsequent works, from the point of view of axiomatization, computation, resistance to strategic behaviour; as far as we know, they have not been studied in the context of runoffs.

\section{Approval-based committee rules}\label{abc-defin}

\par 
Let $\mathcal C = \{c_1, \dots, c_m\}$ be a set of $m$ candidates and $\mathcal V = \{ v_1, \dots, v_n\}$ a set of $n$ voters. An approval profile is a collection of approval ballots $V = \left \langle A_1, \dots, A_n  \right \rangle$ with $A_i \subseteq \mathcal C$ for all $i$. An Approval-based committee rule (ABC rule) $F$ is a rule that takes as input an approval profile $V$ and a committee size $k \ge 1$ and return a set of winning committees $W \subseteq \mathcal C$ of size $|W| = k$.
\par 
For an approval profile $V$ we denote $S_V(c) = \left | \{i | c \in A_i \} \right |$ the approval score of a candidate $c \in \mathcal C$. By extension, the approval score of a set of candidates $J \subseteq \mathcal C$ is the number of approval ballots that contains all candidates from the set $J$, $S_V(J) = \left | \{i | J \subseteq A_i \} \right |$. For simplicity, we write sets on a simpler form, e.g. $S_V(abc)$ instead of $S_V(\{a,b,c\})$. We call \emph{approval winners} the candidates that maximize $S_V$, that is the winners of standard (single-winner) approval voting.

\subsection{Rules}\label{sec:rules}
\par 
As we explained in the introduction, an approval with runoff (AVR) rule use an ABC rule to select two finalists and then return the majority winner between the two finalists. \citep{LacknerS20} did an extensive study of these ABC rules. Because it is the case that interests us, in this section we define rules in the case of a committee size of $k=2$. We will see that most of these rules can be defined with a simple formula for this particular case.
\par 


The most intuitive rule is probably the one that selects the candidates with the highest approval scores:
\smallskip

\drule{Multi-Winner Approval Voting (MAV)}{
$$MAV(V) = \argmax_{x_1, x_2 \in \mathcal C} S_V(x_1) + S_V(x_2)$$
}

\par 
Some rules discount the satisfaction of
voters who are already satisfied by one of the two finalists. This is the case of Proportional Approval Voting (PAV) and Approval Chamberlin Courant (CCAV). In PAV, a voter approving $j$ candidates of the committee $W$ gives a score $s= \sum_{i=1}^{j} \frac{1}{j}$ to the committee. For a committee of size $k =2$, this means that a voter approving one candidate give a score of $1$, and a voter approving both candidates a score of $\frac{3}{2}$. If we do $S_V(x_1) + S_V(x_2)$ as in MAV, voters approving both candidate give a score of $2$ instead of $\frac{3}{2}$, so we have to discard a score of $\frac{1}{2}$ for each of these voters. In CCAV, a voter approving one or more candidates of the committee gives a score of $1$ to the committee. Therefore, in comparison to MAV, we have to discard $1$ on the score given by voters approving both candidates. This is why in the case $k=2$, we can write these rules like this:

\smallskip

\drule{Proportional Approval Voting (PAV)}{
$$PAV(V) = \argmax_{x_1, x_2 \in \mathcal C} S_V(x_1) + S_V(x_2) - \frac{1}{2}S_V(x_1x_2)$$}

\drule{Approval Chamberlin Courant (CCAV)}{
$$CCAV(V) = \argmax_{x_1, x_2 \in \mathcal C} S_V(x_1) + S_V(x_2) - S_V(x_1x_2)$$
}

\par 
These rules select the pairs of candidates $\{x_1, x_2\}$ maximizing $S_V(x_1) + S_V(x_2) - \alpha S_V(x_1x_2)$ for some $\alpha \in [0,1]$. This $\alpha$ is equal to $0$ for MAV, to $1/2$ for PAV and to $1$ for CCAV. We call these rules $\alpha$-AV rules. There also exists sequential versions of these rules. In these sequential versions, the first finalist is always an approval winner.
\smallskip

\drule{Sequential Proportional Approval Voting (S-PAV)}{The rule
chooses the pairs $\{x_1, x_2\}$ such that $x_1$ maximizes $S_V(x_1)$ and $x_2$ maximizes
$S_V(x_2) - \frac{1}{2}S_V(x_1x_2)$.
}
\smallskip

\drule{Sequential Approval Chamberlin Courant (S-CCAV)}{The rule chooses the pairs $\{x_1, x_2\}$ such that $x_1$ maximizes $S_V(x_1)$, and $x_2$ maximizes $S_V(x_2) - S_V(x_1x_2)$.
}
\par 
As before, these definitions of S-PAV and S-CCAV for the case $k=2$ are equivalent to the definitions in the general case. In S-PAV, every voter that approves the first finalist $x_1$ now has a weight of $\frac{1}{2}$, and in S-CCAV, it has a weight of $0$.
Note that sequential MAV would be equivalent to standard MAV. For these sequential rules, the
first finalist is an approval winner $x_1$, and the second finalist maximizes the value $S_V(x_2) - \alpha S_V(x_1x_2)$ for some $\alpha \in [0,1]$. We call these rules $\alpha$-seqAV rules.

\par 
A rule that almost falls into this family is the Eneström Phragmen rule. This rule is also sequential, as the first finalist $x_1$ is an approval winner. Then, the weight of voters approving $x_1$ is reduced to $\max(0, 1-\frac{Q}{S_V(x_1)})$ where $Q$ corresponds to some quota $Q \in [0, n]$. Most of the time, we use Droop quota $Q = n/(k+1)$ which is $n/3$ in our case, or Hare quota $Q= n/k$, equal to $n/2$ in our case. Thus, in the case $k=2$, we can define this rule as follows:

\smallskip

\drule{Eneström Phragmen (EnePhr) }{The rule chooses the pairs $\{x_1, x_2\}$ such that $x_1$ maximizes $S_V(x_1)$ and $x_2$ maximizes $S_V(x_2) - \min(1, \frac{Q}{S_V(x_1)})S_V(x_1x_2)$ for some quota $Q \in [0, n]$.
}

\smallskip
This gives us a $\alpha$-seqAV rule dependent on $Q$, with $\alpha = \min(1, \frac{Q}{S_V(x_1)})$.
\par 
Another popular rule to obtain proportionality in approval based committee selection is the sequential Phragmen rule.

\drule{Sequential Phragmen (S-Phr)}{The rule chooses the pairs $\{x_1, x_2\}$ such that $x_1$ maximizes $S_V(x_1)$ and $x_2$ minimizes
\[
\frac{1 + \frac{S_V(x_1,x_2)}{S_V(x_1)}}{S_V(x_2)}
\]
}
\par
To see how we obtain this simple formulation, we have to use the discrete formulation of S-Phr presented in \citep{LacknerS20}. The load of every voter is initialized at $l_i = 0$. The first finalist selected $x_1$ is an approval winner. Then every voter that approves $x_1$ get a load of $l_i = \frac{1}{S_V(x_1)}$. The second finalist is the candidate $x_2$ minizing:
\[
\frac{1 + \sum_{x_2 \in A_i} l_i}{S_V(x_2)} = \frac{1 + \sum_{x_1, x_2 \in A_i} \frac{1}{S_V(x_1)}}{S_V(x_2)} = \frac{1 + \frac{S_V(x_1x_2)}{S_V(x_1)}}{S_V(x_2)}
\]
\par 
In the next rule, every voter has a weight of $1$ and splits its weight between the candidates he approves:

\drule{Splitted Approval Voting (SAV)}{The rule chooses the pairs $\{x_1, x_2\}$ such that $x_1$ and $x_2$ maximize the splitted approval score $Sp_V(x)$ with
\[
Sp_V(x) = \sum_{i : x\in A_i} \frac{1}{|A_i|}
\]
}
That means that each voter gives the same fraction of vote to each candidate he supports.
\par

\begin{example}\label{ex:ex1}
\begin{table}[!t]
    \centering
    \begin{tabular}{|c|c|c|c|c|c|c|} \hline
    Rule & MAV & (S-)PAV & (S-)CCAV & S-Phr & SAV  \\\hline
     $\{a,b\}$  & $12 + 10 =$\textbf{22} & $12 + 10 - \frac{10}{2} =$ 17 & $12 + 10 - 10 =$ 10 &
     $\frac{1 + 10/12}{10} =$ 22/120 & $2 + 6 + 4\frac{2}{3} =$\textbf{10.7}\\
     $\{a, c\}$ & $12 + 8 =$ 20 & $12 + 8 - \frac{4}{2} =$ \textbf{18} & $12 + 8 - 4 =$ 16 & $ \frac{1 + 4/12}{8}$ = \textbf{16/96} & $2 + 6\frac{1}{2} + 4\frac{2}{3} + 4\frac{1}{2} =$ 9.7\\
     $\{a, d\}$ & $12 + 5 =$ 17 & $12 + 5 - \frac{0}{2} =$ 17 & $12 + 5 - 0 =$ \textbf{17} & $ \frac{1 + 0/12}{5}$ = 1/5 & $2 + 10\frac{1}{2} + 4\frac{1}{2} + 1 =$ 10\\\hline
    \end{tabular}
    \caption{Score of the different committees in Example \ref{ex:ex1}}
    \label{tab:ex1}
\end{table}
Let $V = (2 \times a, 6 \times ab, 4 \times abc, 4 \times cd, 1 \times d)$, i.e., two ballots $\{a\}$, six $\{a,b\}$ etc. Table \ref{tab:ex1} summarizes the score of the three main committees $\{a,b\}$, $\{a,c\}$ and $\{a,d\}$ for the different voting rules. For sequential rule S-PAV and S-CCAV, they are here equivalent to their non-sequential versions, because the approval winner is $a$. For Eneström Phragmen rule, with the Droop quota $Q = n/(k+1) = 17/3$, we have $\alpha = \frac{Q}{S_V(a)} = \frac{17}{3 \times 12} = 17/36$ which is very close to $1/2$, thus the results will be similar to S-PAV. The interesting case is with Hare quota $Q = n/2 = 17/2$, we have $\alpha = \frac{Q}{S_V(a)} = \frac{17}{2 \times 12} = 17/24$. This gives the score of $12 + 8 - 4\frac{17}{24} = $ 17.2 for $\{a,c\}$ and $12 + 5 - 0 =$ 17 for $\{a,d\}$. Thus, $\{a,c\}$ is also the pair of finalists with Hare quota.
\end{example}

We also need the rule that returns all pairs of candidates:\smallskip

\drule{Trivial Approval Voting (TRIV)}{\\
\centerline{$TRIV(V) = \{ \{x,x'\} \mid x,x' \in C, x \ne x'\}$}
}
\par
Note that the trivial approval rule with runoff is actually not completely trivial: it outputs all candidates except the Condorcet loser whenever there is one.

\subsection{Axioms}

\par
\cite{LacknerS20} already did an extensive study of these approval based committee rules. We only complete it here by adding one property that is very interesting for us: \emph{favorite-consistency}. It says that at least one of the candidates in the committee should be an approval winner. This properties is important because it is hard for voters to accept a voting rule in which the candidate with the highest amount of approval is not sure to go to the second round. This property and all the properties of this section are defined for every committee size $k \ge 1$.
\begin{definition}
An ABC rule $F$ is said to be \textbf{favorite-consistent} if every winning committee contains an approval winner, i.e. for all $ W \in F(V)$, we have $W \cap \argmax_{c \in \mathcal C} S_V(c) \ne \emptyset$.
\end{definition}
\par 
Among the voting rules considered in Section \ref{sec:rules}, only sequential rules satisfy this property.

\section{Approval with Runoff}\label{avr-defin}

\subsection{The model}
\par 
We already define approval profiles as collections of approval ballots. An ordinal preference profile is a collection of rankings $\succ = \left < \succ_1, \dots, \succ_n \right >$, where $\succ_i$ is the preference ranking of voter $i$ over $\mathcal C$. An \emph{approval-preference} profile is a collection of pairs $P = \left \langle (A_1, \succ_1), \dots (A_n, \succ_n) \right \rangle$ where $V_P = \left \langle A_1, \dots, A_n \right \rangle$ is an approval profile and $\succ = \left \langle \succ_1, \dots, \succ_n \right \rangle$ an ordinal preference profile.
We also note $P = (V_P, \succ)$. 
\par 
In this section, we assume {\em ballot consistency}:
voter $v_i$ has a threshold in her ranking $\succ_i$ such that every candidate above the threshold is approved and every candidate below is not; formally, $a \succ_i b$ holds for all $a \in A_i$ and $b \not \in A_i$. Ballot consistency allows us to use the following notation \citep{BramsS09}: $x_1 x_2 \ldots x_j | x_{j+1} \ldots x_m$ represents $(\succ_i, A_i)$ with $x_1 \succ_i x_2 \succ_i  \ldots \succ_i x_m$ and $A_i = \{x_1, \ldots, x_j\}$.
\footnote{Ballot consistency does not necessarily hold if voters are strategic and cast insincere approval ballots. Most results in the paper still hold without assuming ballot consistency.}
\par 
Given an ordinal preference profile $\succ$, $maj(\succ, \{a,b\})$ is defined as the set of winners of the majority vote between $a$ and $b$ (which is a singleton except in the case of a tie).
\par 
We now define the family of approval with runoff (AVR) rules. The idea is that we use the approval ballots in the first round to select two finalists, and the second round consists in a majority vote between the two selected candidates. Let $F$ be an (irresolute) approval-based 2-committee rule, i.e. a function that takes as input an approval profile and returns a nonempty set of pairs of candidates. Then, $F^R$ is the (irresolute) AVR rule such that we conduct the majority rule on every pair of finalists selected by $F$. Formally:
\[
F^R(V, \succ) = \bigcup_{\{x,y\} \in F(V)} maj(\succ,\{x,y\}) 
\]





\begin{example}
We continue Example \ref{ex:ex1} where $V = (2 \times a, 6 \times ab, 4 \times abc, 4 \times cd, 1 \times d)$. We can define the approval-preference profile $P = (V, \succ) = (2 \times a|bcd, 3 \times ba|dc, 3 \times ab|dc, 4 \times bac|d, 2 \times cd|ba 2 \times dc|ba, 1 \times d|bac)$, i.e. 2 voters approving $\{a\}$ with ranking $a > b > c > d$, 3 voters approving $\{b,a\}$ with ranking $b >a > d > c$, etc. In this profile $maj(\succ, \{a,b\}) = b$, $maj(\succ, \{a,c\}) = a$ and $maj(\succ, \{a,d\}) = a$. Thus, for $F$ equal to MAV and SAV, $F^R(P) = maj(\succ, \{a,b\}) = b$, for $F$ equal to (S-)PAV, S-Phr and EnePhr, $F^R(P) = maj(\succ, \{a,c\}) = a$ and for $F$ equal to (S-)CCAV, $F^R(P) = maj(\succ, \{a,d\}) = a$.
\end{example}

\subsection{Axiomatic analysis}
\label{normative}

In this section, we study the axiomatic properties of AVR rules. 
\par

As for ABC rules, an AVR rule $F^R$ is {\em anonymous} if it is invariant by any permutation of the voters, and {\em neutral} if for any permutation of the candidates $\pi$ and every profile $P$, $F^R(\pi(P)) = \pi(F^R(P))$.

\par
We will use the following unanimity condition, that is a strengthening of strict Pareto, adapted to the approval-preference case. We say that candidate $a$ unanimously preference-approval dominates candidate $b$ if
\begin{enumerate}
    \item for every voter $v_i$, $a \succ_i b$
    \item for some voter $v_i$, $a \in A_i$ and $b \not \in A_i$
\end{enumerate}

Together with ballot consistency, it implies that every voter who approves $b$ also approve $a$, and at least one voter who approves $a$ does not approve $b$. For simplicity we refer to this condition as our Pareto condition, and say that $a$ dominates $b$ when $a$ unanimously preference-approval dominates $b$.

\begin{definition}
An AVR rule $F^R$ is \textbf{Pareto-efficient} if for all approval-preference profile $P$ in which there exists $a,b \in \mathcal C$ such that $a$ dominates $b$, we have $b \not \in F^R(P)$.
\end{definition}

\begin{proposition}
AV$^R$, S-PAV$^R$, S-Phr$^R$, EnePhr$^R$, PAV$^R$ and SAV$^R$ are Pareto-efficiency, but not S-CCAV$^R$, CCAV$^R$ and TRIV$^R$.
\end{proposition}
\begin{proof}
Consider a profile $P$ in which a candidate $a$ dominates another candidate $b$. Let $F$ be an approval based committee rule and $F^R$ the approval with runoff rule associated to it. We want to show that $b \notin F^R(P)$. We assume by contradiction that $b \in F^R(P)$. Thus, there exists $x \in \mathcal C$ such that $\{b,x\} \in F(V)$. 
\par 
 $x \ne a$, because we know that $maj(\succ, \{a,b\}) = a$ since $a$ dominates $b$. We are going to show that for every rule $F \in \{$MAV, S-PAV, S-Phr, PAV, SAV$\}$, we have a contradiction.
\begin{itemize}
    \item Let $F$ be an $\alpha$-AV rule with $\alpha < 1$ (like MAV and PAV). Define $f(y,y')$ the $\alpha$-AV score of a pair of candidates $\{y,y'\}$: 
    \begin{align*}
        f(y,y') &= S_V(y) + S_V(y') - \alpha S_V(yy') \\
        &=S_V(y) + (1-\alpha)S_V(y') + \alpha(S_V(y') - S_V(yy'))
    \end{align*}
    Since $a$ dominates $b$, we have $S_V(a) > S_V(b)$ and for all $x$,  $S_V(a)-S_V(ax) \ge S_V(b)-S_V(bx)$. We have $f(x,a) = S_V(a) + (1-\alpha)S_V(x) + \alpha(S_V(a) - S_V(xa)) > S_V(b) + (1-\alpha)S_V(x) + \alpha(S_V(b) - S_V(xb)) = f(x,b)$, thus $\{x,b\}$ cannot be a pair of finalists because the rule selects the pairs of candidates maximizing $f$.
    \item Le $F$ be an $\alpha$-seqAV rule with $\alpha < 1$ (like MAV and S-PAV). $b$ is not the first finalist selected because $S_V(a) > S_V(b)$. Let $x$ be the first finalist selected and 
    \begin{align*}
        f_x(y) &= S_V(y) - \alpha S_V(xy)\\
        &= (1-\alpha)S_V(y) + \alpha(S_V(y) - S_V(xy))
    \end{align*}
    Again, $a$ dominates $b$, so we have $S_V(a) > S_V(b)$ and for all $x$,  $S_V(a)-S_V(ax) \ge S_V(b)-S_V(bx)$. We have $f_x(a) =  (1-\alpha)S_V(a) + \alpha(S_V(x) - S_V(xa))> (1-\alpha)S_V(b) + \alpha(S_V(x) - S_V(xb)) = f_x(b)$, thus $\{x,b\}$ cannot be a pair of finalists because the rule selects the candidates maximizing $f_x$.
    \item Let $F =$ S-Phr. Since $S_V(a) > S_V(b)$, $b$ is the second finalist selected and $x$ is the first finalist. For $y \in \mathcal C$, let
    \begin{align*}
        f_x(y) &= \frac{1 + \frac{S_V(xy)}{S_V(x)}}{S_V(y)} \\
        &= \frac{1}{S_V(x)} \frac{S_V(x) + S_V(xy)}{S_V(y)}
    \end{align*}
    We want to prove that $\frac{S_V(x) + S_V(xa)}{S_V(a)} < \frac{S_V(x) + S_V(xb)}{S_V(b)}$. Now, observe that 
    \begin{align*}
    \frac{S_V(x) + S_V(xb)}{S_V(b)} \ge \frac{S_V(x)}{S_V(b)} > 1 \ge \frac{S_V(xa)-S_V(xb)}{S_V(a) - S_V(b)}
    \end{align*}
    which gives
    \begin{align*}
    &(S_V(x) + S_V(xb))(S_V(a) - S_V(b)) > (S_V(xa) - S_V(xb))S_V(b) \\
    \Leftrightarrow & (S_V(x) + S_V(xb))S_V(a) + S_V(b)(S_V(x)-S_V(xb)) > (S_V(xa) + S_V(x))S_V(b) + S_V(b)(S_V(x)-S_V(xb)) \\
    \Leftrightarrow & (S_V(x) + S_V(xb))S_V(a) > (S_V(xa) + S_V(x))S_V(b) \\
     \Leftrightarrow & \frac{S_V(x) + S_V(xb)}{S_V(b)} > \frac{S_V(x) + S_V(xa)}{S_V(a)} 
    \end{align*}
    Therefore, $f_x(a) < f_x(b)$, thus $\{b,x\}$ cannot be a pair of finalists because S-Phr selects the candidates minimizing $f_x$.
    \item If $F = $ SAV, we define the Splitted Approval score $Sp_V(y) = \sum_{y \in A_i} \frac{1}{|A_i|}$. The two finalists are $x$ and $b$. However, $Sp_V(a) = Sp_V(b) + \sum_{i:a \in A_i \text{ and } b \notin A_i}\frac{1}{|A_i|} > Sp_V(b)$. This contradicts that $\{b,x\}$ is a possible pair of finalists.
\end{itemize}
In all those cases, we showed that $\{b,x\} \notin F(V)$ for all $x$. Therefore, $b \notin F^R(P)$.
\par 
To show that the other rules are not Pareto-efficient, consider the profile $P = (V, \succ) = (abc|, ab|c, a|bc, k \times |bca)$ with $k > 3$. For $F \in \{\text{CCAV}, \text{S-CCAV}\}$, $F(V) = \{\{a,b\}, \{a,c\}\}$ and $maj(\succ, \{a,c\}) = c$, so $c \in F^R(V)$. However, $b$ dominates $c$. Thus, $F$ is not Pareto-efficiency. 
\par 
For EnePhr with quota $Q > 0$, observe that it is equivalent to S-CCAV if $\frac{Q}{S_V(a)} \ge 1$. Then, if $Q = \beta n = \beta (k+3)$ with $\beta \in [0,1]$, we set $k$ in $P$ such that $\frac{Q}{S_V(a)} \ge 1$, i.e. $k \ge 3(\frac{1}{\beta} - 1)$.
\end{proof}

\par 

\par 
\emph{A $i$-deviation} of a profile $P$ is a profile $P'$ such that for all $j \ne i, A_j = A'_j$ and $\succ_j = \succ'_j$. We define strong strategy-proofness as the impossibility for a voter to deviate from a profile to another in which there is a winner that she prefers to all the winners of the first profile. We also define weak strategy-proofness as the impossibility for a voter to deviate from a profile where she  does not approve any winning candidate to one where she approves at least one winning candidate.

\begin{definition}
An approval with runoff rule $F^R$ is \textbf{strongly strategy-proof} if for every profile $P$, there is no $i$-deviation $P'$ of $P$ such that $\exists x \in F^R(P'), \forall y \in F^R(P), x \succ_i y$. 
\end{definition}

\begin{definition}
An AVR rule $F^R$ is \textbf{weakly strategy-proof} if for every profile $P$, there is no $i$-deviation $P'$ of $P$ such that $\left | F^R(P) \cap A_i \right | = 0$ and $\left | F^R(P') \cap A_i \right | \ge 1$. 
\end{definition}

It is hard to get strategy-proofness in approval-based committee voting, as it is incompatible with proportionality \citep{DBLP:journals/corr/abs-2104-08594}. 
For AVR rules, this is no better: weak strategy-proofness is incompatible with Pareto-efficiency. Among the rules defined in Section \ref{sec:rules}, the only strategy-proof rule is TRIV$^R$.

\begin{theorem} \label{thm:strategy-proof}
No AVR rule is weakly strategy-proof and Pareto-efficient.
\end{theorem}
\begin{proof}
\begin{table}[!t]
    \centering
    \begin{tabular}{|c|c|c|c|c|} \hline
      &  $v_1$& $v_2$& $v_3$  & Finalists\\ \hline 
      $V = V'$ & $\{a\}$ & $\{a\}$ & $\{c\}$ & $\{a,c\}$ \\ \hline  
      $V''$ & $\{a,b\}$ & $\{a\}$ & $\{c\}$ &  $\{a,c\}$\\ \hline 
      $V^* $ & $\{a,b\}$ & $\{a\}$ & $\{b,c\}$ & $\{a,b\}$ \\ \hline
    \end{tabular}
    \caption{Profiles used in proof of Theorem \ref{thm:strategy-proof}}
    \label{tab:proof_thm_stgyproof}
\end{table}
Assume $F^R$ is Pareto-efficient and weakly strategy-proof. 
Let $P = (V, \succ) = (a|bc, a|bc, c|ab, 10 \times abc|)$ --- that is, $V$ contains 2 approval ballots $\{a\}$, 1 $\{c\}$ and 10 $\{a,b,c\}$, and $\succ$ contains 12 rankings $a \succ b \succ c$ and 1 $c \succ a \succ b$. In $P$, $a$ dominates $b$ and $maj(\succ, \{b,c\}) = b$, which implies $\{c,b\} \not \in F(V)$. Let $P' = (V', \succ) = \{a|cb, a|cb, c|ab, 10 \times cba|\}$. In $P'$, $c$ dominates $b$ and $maj(\succ, \{a,b\}) = b$, therefore  
$\{a,b\} \not \in F(V')$. Because $V = V'$, $F(V) = F(V') = \{a,c\}$. 

\par 

Let $P'' = (V'', \succ'') = (ab|c, a|bc, c|ab, 10 \times abc|)$. $a$ dominates $b$ in $P''$, and $maj(\succ'',\{b,c\}) = b$. By Pareto-efficiency, $b \notin F^R(P'')$ and thus $\{b,c\} \not \in F(V'')$. 
Assume $\{a,b\} \in F(V'')$. Let $\overline{P} = (V, \overline{\succ}) = (a|bc, a|bc, c|ab, 10 \times cab|)$ and a deviation $\overline{P''} = (V'', \overline{\succ}) = (a\underline{b}|c, a|bc, c|ab, 10 \times cab|)$, where voter $v_1$ deviates from $\{a\}$ to $\{a,b\}$. We know $F(V) = \{a,c\}$ so $F^R(\overline{P}) = maj(\overline{\succ}, \{a,c\}) = c$, and $a \in F^R(\overline{P''})$ because we assumed $\{a,b\} \in F(V'')$. This is a successful manipulation, which contradicts strategy-proofness. Thus, $\{a,b\} \notin F(V'')$ and $F(V'') = \{a,c\}$.

\par 

Let $P^* = (V^*, \succ^*) = (ab|c, a|bc, bc|a, 10 \times bca|)$. In $P^*$, $b$ dominates $c$ and $maj(\succ^*, \{a,c\}) = c$. Therefore, $\{a,c\} \notin F(V^*)$. Assume 
that $\{b,c\} \in F(V^*)$. Let $\widehat{P} = (V, \widehat{\succ}) = (10 \times acb|, ab|c, a|bc, c|ba)$ and $\widehat{P}^* = (V^*, \widehat{\succ}) = (ab|c, a|bc, c\underline{b}|a, 10 \times acb|)$ a deviation of voter $v_2$. We have $F^R(\widehat{P}) = maj(\widehat{\succ}, \{a,c\}) = a$. However, if we assume $\{b,c\} \in F(V^*)$, then $c \in F^R(\widehat{P}^*)$ because $maj(\widehat{\succ}, \{b,c\}) = c$. Therefore, the deviating voter makes $c$ win. This contradicts strategy-proofness, therefore, $\{b,c\} \notin F(V^*)$ and $F(V^*) = \{a,b\}$.

\par 

Finally, let $\breve{P^*} = (V^*, \breve{\succ}) = (ab|c, a|bc, cb|a, 10 \times cab|)$ and the deviation $\breve{P} = (V, \breve{\succ}) = (, ab|c, a|bc, c|\underline{b}a, 10 \times cab|)$,
where the voter $v_3$ deviates. We have $F^R(\breve{P^*}) = maj(\breve{\succ},\{a,b\}) = a$ and $F^R(\breve{P}) = maj(\breve{\succ},\{a,c\}) = c$. This deviation is a manipulation. This contradicts strategy-proofness, and proves the theorem.

\end{proof}

This set of properties is minimal: TRIV$^R$ is weakly strategy-proof but not Pareto-efficient, MAV$^R$ is Pareto-efficient but not weakly strategy-proof. 
\begin{proposition}
TRIV$^R$ is strategy-proof but not AV$^R$, S-PAV$^R$, S-Phr$^R$, EnePhr$^R$, S-CCAV$^R$, PAV$^R$, CCAV$^R$ and SAV$^R$.
\end{proposition}
\begin{proof}
We know by Theorem \ref{thm:strategy-proof} that no rule is Pareto-efficient and Strategy-Proof. Thus, AV$^R$, (S-)PAV$^R$, S-Phr$^R$, EnePhr$^R$, and SAV$^R$ are not weakly strategy-proof. The trivial rule is clearly strategy-proof. For CCAV$^R$ and S-CCAV$^R$, consider the profile $P = (V, \succ) = (ca|b, c|ab, b|ca, 10 \times abc|)$. In this profile, with these two rules the finalists are $\{c, b\}$ and the winner is $b$. If the voter $v_1$ approving $\{c,a\}$ now approves $\{a\}$ only, we obtain the profile $P' = (V', \succ') = (a|\underline{c}b, c|ab, b|ca, 10 \times abc|)$ and all pairs of finalists are possible. Thus, $a$ is a possible winner as $maj(\succ', \{a,b\}) = a$. This is a manipulation.
\end{proof}
\par We now focus on monotonicity.
Given a profile $P =(V,\succ)$ and $a \in \mathcal C$, a profile $P' = (V', \succ') \ne P$ is an $a$-improvement of $P$ if for some $i \in N$ we have
\begin{enumerate}
    \item $A'_i = A_i \cup \{a\}$ or $A'_i = A_i$
    \item For all $x,y$ with $y \ne a$, if $x \succ_i y$ then $x \succ'_i y$
    \item For all $ j \ne i, A'_j = A_j$ and $\succ'_j = \succ_j$
\end{enumerate}
\begin{definition}
An AVR rule $F^R$ is \textbf{monotonic} if for every $a \in F^R(P)$ and for every $a$-improvement $P'$ of $P$, we have $a \in F^R(P')$.
\end{definition}

\begin{proposition}
AV$^R$ and TRIV$^R$ is monotonic but not S-PAV$^R$, S-Phr$^R$, EnePhr$^R$, S-CCAV$^R$, PAV$^R$, CCAV$^R$, SAV$^R$ are not.
\end{proposition}
\begin{proof}
Let's show that MAV$^R$ is monotonic. Let $P = (V, \succ)$ be a profile in which $a \in \mavr{(P)}$. Let $x \in \mathcal C$ such that $\{a,x\} \in F(V)$. Let $P' = (V', \succ')$ be an $a$-improvement of $P$. Since $S_{V'}(a) = S_{V}(a)+ 1$ and for all$ y \ne a$, $S_{V'}(y) = S_V(y)$, it is obvious that $\{a, x\} \in F(V')$. Since $a \in maj(\succ, \{a,x\})$, it is also obvious that $a \in maj(\succ', \{a,x\})$ and thus $a \in \mavr{(P)}$.
\par 
To show that the other rules are not monotonic, we can use a counter example. Consider the profile $P = (V, \succ) = (a|bc, a|bc, b|ca, c|ba, 10 \times cab|)$. In this profile, for all rule $F \in \{$CCAV, PAV, S-CCAV, S-PAV, EnePhr, S-Phr, SAV$\}$, $F(V) = \{\{a,b\}, \{a,c\}\}$ and $F^R(P) = \{a, c\}$. $a$ is a winner of the election. Now consider the $a$-improvement $P' = (V', \succ') = (a|bc, a|bc, b\underline{a}|c, c|ba, 10 \times cab|)$. Now, for all the rules considered before, $F(V) = \{c,a\}$ and $F^R(P) = \{c\}$ and $a \notin F^R(P)$. Thus, none of these rules is monotonic.
\end{proof}

MAV$^R$ satisfies both monotonicity and Pareto-efficiency. The only other rule that we considered in Section \ref{sec:rules} that is monotonic is the trivial rule, which is however not Pareto-efficient. However, MAV$^R$ is not the only rule satisfying these two properties. This is also the case for $F^R$, where $F$ returns all pairs of candidates $\{a,b\}$ such that either (i) neither $a$ nor $b$ is Pareto-dominated in $V$ or (ii) $a$ is the only candidate that dominates $b$ in $V$. Indeed, let $b$ be a candidate dominated by $a$ in some profile $P = (V, \succ)$. Then, if there exists $x$ such that $\{x, b\} \in F(V)$, we must have $x = a$ by definition of $F$. Since $a$ dominates $b$, it is clear that $maj(\succ, \{a,b\}) = a$, and $b\notin F^R(P)$. This proves Pareto-efficiency.  
\par
Let's now show monotonicity. For a profile $P = (V, \succ)$, let $a \in F^R(P)$ a winning candidate, and $P'$ be a $a$-improvement of $P$. Let $\{x,a\} \in F(V)$ be the pair of finalists such that $maj(\succ,\{x,a\}) = a$. $x$ is either not dominated or dominated only by $a$. If one new voter approves $a$, then $x$ is still either not dominated or dominated only by $a$, therefore we still have $\{x,a\} \in F(V)$. If no new voter approves $a$, then we obviously have $\{x,a\} \in F(V)$. If $a$ wins the majority vote against $x$ in $P$, then it clearly wins it in $P'$. Therefore, $F$ is also monotonic and Pareto-efficient (but also neutral and anonymous), which means we cannot characterize MAV$^R$ with only these properties.


\par 
Finally, we focus on clone-proofness. Informally, this property means that adding a clone of a candidate 
does not change significantly the outcome of the election. Formally, let $a \in \mathcal C$, and a profile $P'$ over set of candidates $C' = C \cup \{a'\}$. $P'$ is an $a$-cloning extension of $P$ if
\begin{enumerate}
    \item For every $v_i$, $a \in A_i$ if and only if $a' \in A_i$
    \item For every $v_i$ and $x \not \in \{a, a'\}$, $a \succ_i x$ if and only if $a' \succ_i x$
\end{enumerate}

\begin{definition}
An AVR rule $F^R$ is \textbf{clone-proof} if for any profile $P$, candidate $a \in \mathcal C$, and an $a$-cloning extension $P'$ of $P$, the two following conditions hold:
\begin{enumerate}
    \item For every $c \ne a$, $c \in F^R(P)$ if and only if $c \in F^R(P')$
    \item $a \in F^R(P)$ if and only if  $\left | F^R(P') \cap \{a,a'\} \right | \ge 1$
\end{enumerate}
\end{definition}

$\mavr$ is not clone-proof. Let $P = (\succ, V) = (a|b, ba|, ba|)$;  $\mavr(P) = \{b\}$. 
$P' = (V', \succ') =  (aa'|b, baa'|, baa'|)$ is an $a$-cloning extension of $P$ and yet $ \mavr(P) = maj(\succ', \{a,a'\}) = \{a\}$.
At first sight, $\ccavr$ and $\sccavr$ seem clone-proof, but $P$ can be used to prove that $\ccavr$ and $\sccavr$ are not clone-proof either:
$\ccavr(P)=\sccav(P)= \{b\}$; $P' = $ $(V', \succ')= (aa'|b, baa'|, baa'|)$ is an $a$-cloning extension of $P$; and yet $\ccav(V') =\sccav(V') = \{(b,a),(b,a'),(a,a')\}$, so $\ccavr(P')=\sccavr(P')$ contains $a$, 
breaking clone-proofness.

\par 
Among the rules considered in Section \ref{sec:rules}, none is actually clone-proof. However, there exist AVR clone-proof rules: such a rule is defined by the ABC rule that selects the pairs of candidates maximizing $f(x_1,x_2) = S_V(x_1) + S_V(x_2) - 2 S_V(x_1x_2)$.
However, this rule is not Pareto-efficient. More generally, clone-proofness and Pareto-efficiency are incompatible:

\begin{theorem}
No AVR rule is clone-proof and Pareto-efficient.
\end{theorem}
\begin{proof}
Assume $F^R$ is clone-proof and Pareto-efficient.
Take the profile $P = (V, \succ) = (a|b, ba|, ba|)$ and let $P' = (V', \succ') = (aa'|b, baa'|, baa'|)$ be an $a$-cloning extension of $P$. Because $F^R(P) = \{b\}$, by clone-proofness we have $F^R(P') = \{b\}$, therefore, $\{a,b\} \in F(V')$ or $\{a',b\} \in F(V')$. Without loss of generality, assume $\{a,b\} \in F(V')$. Let $P^* = (V^*, \succ^*) = (aa'|b, a'ba|, a'ba|)$ in which $a'$ dominates $b$. Since $V^* = V'$, $\{a,b\} \in F(V^*)$ and thus $b \in F^R(P^*)$, which contradicts Pareto-efficiency.
\end{proof}

\par 
We now define a weaker version of clone-proofness, with a domain restriction that eliminates pathological profiles:
\begin{definition}
An AVR rule $F^R$ is \textbf{weakly clone-proof} if it is clone-proof on every profile $P$ such that no candidate $c\in \mathcal C$ is approved in every non-empty ballot: for every $c \in \mathcal C$, there exists a voter $v_i$ such that $A_i \ne \emptyset$ and $c \not \in A_i$. 
\end{definition}
CCAV$^{R}$ and S-CCAV$^{R}$ are weakly clone-proof (Proposition \ref{thm:rule_cloneproof}). Recall that they are not Pareto-efficient. A rule weakly clone-proof and Pareto-efficient is $F^R$, where $F$ selects the CCAV finalists, and uses MAV as a tie-breaking if there are several pairs of finalists.
\par 
Unfortunately, we have the following impossibility:

\begin{theorem} \label{thm:cloneproof}
No AVR rule is monotonic, weakly clone-proof.
\end{theorem}

\begin{proof}

\begin{table}[!t]
    \centering
    \begin{tabular}{|c|c|c|c|c|} \hline
      &  $v_1$& $v_2$& $v_3$  & finalists\\ \hline 
      $V$ &  $\{a\}$ & $\{b\}$ & $\{c\}$ &  $\{a,b\} \in F(V)$\\ \hline 
      $V' = V''$ & $\{a\}$ & $\{a,b\}$ & $\{c\}$ & $\{a,b\} \in F(V')$ \\ \hline 
      $V^*$ & $\{a,b\}$ & $\{a,b\}$ & $\{c\}$ & $\{a,b\} \in F(V^*)$ \\ \hline
      $\widehat{V}$ & $\{a\}$ & $\{a\}$ & $\{c\}$ & $F(\widehat{V}) = \{a,c\}$ \\ \hline
      $\widehat{V}^*$ & $\{a, a^*\}$ & $\{a, a^*\}$ & $\{c\}$ & $\{a,a^*\} \not \in F(\widehat{V}^*)$ \\ \hline
    \end{tabular}
    \caption{Profiles used in proof of Theorem \ref{thm:cloneproof}}
    \label{tab:proof_thm_cloneproof}
\end{table}
Assume $F^R$ is monotonic and weakly clone-proof
Let $P = (V,\succ) = (a|bc, b|ca, c|ab, 10 \times cab|)$. We can assume without loss of generality $\{a,b\} \in F(V)$, so 
$a \in F^R(P)$.  $P' = (V',\succ') = (a|bc, b\underline{a}|c, c|ab, 10 \times cab|)$ is an $a$-improvement of $P$. By monotonicity, $a \in F^R(P')$. Because $maj(\succ',\{a,c\}) = c$, this implies $\{a,b\} \in F(V')$.
\par 
Let $P'' = (V'', \succ'') = ( a|bc, ba|c, c|ab, 10 \times cba|)$. Since $V'' = V'$ we have $\{a,b\} \in F(V'')$, and $b \in F^R(P'')$. Let $P^* = (V^*, \succ^*) = (a\underline{b}|c, ab|c, c|ab, 10 \times cba|)$ a $b$-improvement of $P''$. By monotonicity, $b \in F^R(P^*)$. Since $maj(\succ^*, \{b,c\}) = c$, we have $\{a,b\} \in F(V^*)$.
\par 

Now consider $\widehat{P} = (\widehat{V}, \widehat{\succ}) = (a|c, a|c, c|a, 10 \times ca|)$. We have $F^R(\widehat{P}) = \{c\}$. If we clone $a$ into another candidate $b$, we can define the $a$-cloning extension $\widehat{P}^* = (\widehat{V}^*, \widehat{\succ}^*) = (ab|c, ab|c, c|ab, 10 \times cba|)$. By weak clone-proofness, neither $a$ nor $b$ is in $F^R(\widehat{P}^*)$, so $\{a,b\} \notin F(\widehat{V}^*)$. We obtain a contradiction and this concludes the proof.
\end{proof}

This set of properties is minimal: MAV$^R$ is monotonic and neutral; CCAV$^R$ is weakly clone-proof and neutral; a rule with a constant pair of finalists is weakly clone-proof and monotonic.

\begin{proposition}\label{thm:rule_cloneproof}
CCAV$^R$ and S-CCAV$^R$ are weakly clone-proof but not AV$^R$, S-PAV$^R$, S-Phr$^R$, EnePhr$^R$, PAV$^R$ and SAV$^R$ are not.
\end{proposition}
\begin{proof}
Let's show that S-CCAV$^R$ and CCAV$^R$ are weakly clone-proof. Let $P = (V,\succ)$ be a profile in which there is no candidate $c$ that is approved in every non-empty ballot. Let $a \in \mathcal C$ and $P' = (V', \succ')$ be an $a$-cloning extension of $P$ with a clone $a'$ of $a$. Let $f_V(x,y)$ be the Chamberlin-Courant score for two candidates $x$ and $y$ in $V$, i.e. $f_V(x,y) = S_V(x) + S_V(y) - S_V(xy)$. For CCAV, we have for all $ x,y \notin \{a,a'\}, f_{V'}(x,y) = f_V(x,y)$ and $f_{V'}(x,a) = f_{V'}(x,a') = f_V(x,a)$. Finally, $f_{V'}(a,a') = S_V(a) + S_V(a') - S_V(aa') = S_V(a)$. Let $w$ be an approval winner in $V$. Let's prove that $\{a,a'\} \notin (S-)\ccav{(V')}$
\begin{itemize}
    \item If $a$ is not an approval winner in $V$, then $\{a,a'\} \notin \sccav{(V')}$ for sure. For CCAV, we have $f_{V'}(w, a) = f_V(w,a) \ge S_V(w) > S_V(a) = f_{V'}(a,a')$, thus $\{a,a'\} \notin \ccav{(V')}$.
    \item If $a$ is the approval winner in $P$, since there is no candidate $c \in \mathcal C$ that is approved in every non-empty ballot, there exists a candidate $x$ and a $A_i \in V$ such that $x \in A_i$ and $a \notin A_i$. Thus, $f_{V'}(x,a) = f_V(x,a) = S_V(x) + S_V(a) - S_V(xa) \ge S_V(a) + 1 > S_V(a) = f_{V'}(a,a')$, thus $\{a,a'\} \notin \ccav{(V')}$.
\end{itemize} 
In all cases, $\{a,a'\} \notin$ (S-)CCAV($V'$) and the Chamberlain-Courant score of all other pairs is unchanged, so it is easy to observe that for all $ x \notin \{a,a'\}$, $\{x,a'\} \in F(V')$ if and only if $\{x,a\} \in F(V)$ and for all $ x,y \ne a', \{x,y\} \in F(V')$ if and only if $ \{x,y\} \in F(V)$.
Therefore, for all $x \ne a', x \in F^R(P')$ if and only if $x \in F^R(P)$ and $a' \in F^R(P') \Leftrightarrow a \in F^R(P)$. Thus, CCAV is weakly clone-proof.
\par
Let's now show a counter-example for other rules. Consider the profile $P = (V, \succ) = ( a|b, a|b, b|a, 10 \times ba|)$. Note that $P$ satisfies the condition of the definition of weak clone-proofness. Since there are only two candidates, for every rule $F$ we have  $F(V) = \{a,b\}$ and $F^R(P) = maj(\succ, \{a,b\}) = \{b\}$. Let's now clone $a$ into another candidate $a'$. We have the $a$-cloning extension $P' = (V', \succ') = ( aa'|b, aa'|b, b|aa', 10 \times baa'|)$. For all rules $F\in \{$MAV, PAV, S-PAV, S-Phr, SAV\}, we have $\{a,a'\} \in F(V')$ and $a \in F^R(P')$, thus breaking clone-proofness. For EnePhr with a quota $Q$ such that $\frac{Q}{n} <1$, we use the profile $P_k = (V_k, \succ_k) = ( b|a, k \times a|b, k \times ba|)$ with a sufficiently high $k$.
\end{proof}
\par 
\begin{table}[!t]
    \centering
    \begin{tabular}{|c|c|c|c|c|c|c|c|c|c|}
    \hline
        & AV$^{R}$ & S-PAV$^{R}$ & S-Phr$^{R}$& EnePhr$^{R}$ & S-CCAV$^{R}$ &PAV$^{R}$ &  CCAV$^{R}$ &  SAV$^{R}$ & TRIV$^{R}$ \\\hline\hline
        Pareto-efficient & \checkmark & \checkmark & \checkmark & \checkmark &  & \checkmark &  & \checkmark &  \\\hline
        Monotonic & \checkmark & && & & & & & \checkmark \\ \hline
        Strategy-proof & & & && & & & & \checkmark \\ \hline 
        Weakly Clone-proof & & & && \checkmark  && \checkmark & & \\ \hline 
        
    \end{tabular}
    \caption{Approval with Runoff rules and their properties}
    \label{tab:axioms_afterrunoff}
\end{table}

\par 
Finally, Table \ref{tab:axioms_impossibilities} summarizes the impossibilities between the different properties of this section and proposes a voting rule when two properties are possible together.

\begin{table}[!t]
    \centering
    \begin{tabular}{|c|c|c|c|c|c|}
    \hline
       & Strategy-proof & Monotonic & Clone-proof & Weakly clone-proof \\ \hline
       Pareto-efficient &  $\emptyset$ & MAV$^R$ & $\emptyset$ & CCAV$_+^R$ \\ \hline
       Strategy-proof &  & TRIV$^R$ & $\emptyset$ & $\emptyset$\\ \hline
       Monotonic &  && $\emptyset$& $\emptyset$ \\ \hline
       Clone-Proof & & && $2$-AV$^R$ \\ \hline
       
    \end{tabular}
    \caption{Summary of impossibilities between properties (assuming neutrality and anonymity). $\emptyset$ means there is no rule satisfying both properties. Otherwise, we propose a rule satisfying both properties --- note that this is not a characterization.}
    \label{tab:axioms_impossibilities}
\end{table}

\par 
In this table, CCAV$_+^R$ is the rule that select the CCAV winners and solve ties using MAV. $2$-AV is the $\alpha$-AV rule with $\alpha = 2$, i.e. selecting the pairs of finalists $\{x,y\}$ maximizing $f(x,y) = S_V(x) + S_V(y) - 2S_V(xy)$. This way, the score of two clones $a$ and $a'$ is $f(a,a') = 0$. This is the reason why this rule is always clone-proof.
\par 
Let's show that CCAV$_+^R$ is weakly clone-proof and pareto-efficient. For weak clone-proofness, we can directly adapt the proof for CCAV$^R$. For pareto-effiency, consider a profile $P = (V,\succ)$ in which $a$ dominates $b$. It is clear that for all $x \ne a,b$, if $\{x,b\} \in $CCAV$(V)$, then $\{x,a\} \in $CCAV$(V)$. However, $\{x,a\}$ has a better approval score than $\{x,b\}$, so $\{x,b\}$ will never be returned by CCAV$_+^R$ and $b$ cannot be a winner. Thus, the rule is Pareto-efficient.
\par

\section{Statistical Analysis with one-dimensional Euclidean preferences}
\label{sp}

We now want to explore the spectrum between rules that select the most popular candidates, typically MAV, and rules that favour diversity in the set of finalists, typically 
CCAV and S-CCAV or, to a lesser extent, PAV and S-PAV. 

In this section we focus on one-dimensional Euclidean preferences:
we assume that there is a function $\phi : \mathcal V \cup \mathcal C \rightarrow \mathbb R$ 
such that
$c \succ_i c'$ if $|\phi(c) -\phi(v_i)| < |\phi(c') -\phi(v_i)|$. 
Such uni-dimensional preferences are extremely common in the literature and, indeed, capture an important political structure often met in actual politics.

We also need to define the voters approval behavior. Unfortunately, the literature (empirical or theoretical) on this issue is sparse. 
\citep{Green-ArmytageT20} use a threshold for approval that is the average utility but do not offer any justification for that assumption. 
Classical political science \citep{10.2307/2604342, Cox1984StrategicEC, cox_1997} argues that in the first round of a run-off, like in a two-member district, the voter concentrates on which candidate is going to arrive second and third (this is in practice the more pressing question) and will vote on this basis.
More precisely, \citep{LaslierVDS16} put forward theoretical arguments that imply putting the threshold of approval at the utility level of the candidate supposed to arrive second in the race. From an empirical study, \citep{VanderStraetenetal18} conclude that the strategic model behaves well except with respect to candidates who are ex post ranked very low. For models that set individual thresholds of approval, the ones that impose a fixed number of approved candidates work well if this number is of the order of magnitude of the number of seats to be allocated, and the average utility model is the least satisfactory. 

For the simulations in this paper, we will assume that there exists $d > 0$ such that every voter $v_i$ approves all candidates $c$ such that $|\phi(c) -\phi(v_i)| < d$. 
This simple one-parameter family of rules will allow to discuss in a clean manner the collective consequences of voters being more or less flexible in their approvals.

Given a distribution of voters, we want to know in which position a candidate can maximize his score with $\alpha$-seqAV rules.

\begin{figure*}[!t]
    \centering
    \subfloat[Triangular distribution]{
\begin{tikzpicture}

\definecolor{color0}{rgb}{0.12156862745098,0.466666666666667,0.705882352941177}

\begin{axis}[
legend style={fill opacity=0.8, draw opacity=1, text opacity=1, draw=white!80!black},
tick align=outside,
tick pos=left,
title={Distribution of voters},
x grid style={white!69.0196078431373!black},
xmin=-1, xmax=1,
xtick style={color=black},
y grid style={white!69.0196078431373!black},
ymin=0, ymax=0.5,
ytick style={color=black},
width=8.5cm,
height=4.5cm,
]
\path [draw=none, fill=color0]
(axis cs:-1,0)
--(axis cs:0,0.5)
--(axis cs:1,0)
--cycle;
\end{axis}

\end{tikzpicture}
    \label{fig:distrib-0}
    }
    \subfloat[Gaussian distribution]{
\begin{tikzpicture}

\definecolor{color0}{rgb}{0.12156862745098,0.466666666666667,0.705882352941177}

\begin{axis}[
tick align=outside,
tick pos=left,
title={Distribution of voters},
x grid style={white!69.0196078431373!black},
xmin=-2.29779738313012, xmax=2.21805231942632,
xtick style={color=black},
ytick=\empty,
y grid style={white!69.0196078431373!black},
ymin=0, ymax=1387.05,
ytick style={color=black},
width=8.5cm,
height=4.5cm,
]
\draw[draw=none,fill=color0] (axis cs:-2.09253148755937,0) rectangle (axis cs:-2.01042512933107,1);
\draw[draw=none,fill=color0] (axis cs:-2.01042512933107,0) rectangle (axis cs:-1.92831877110277,0);
\draw[draw=none,fill=color0] (axis cs:-1.92831877110277,0) rectangle (axis cs:-1.84621241287448,1);
\draw[draw=none,fill=color0] (axis cs:-1.84621241287448,0) rectangle (axis cs:-1.76410605464618,3);
\draw[draw=none,fill=color0] (axis cs:-1.76410605464618,0) rectangle (axis cs:-1.68199969641788,2);
\draw[draw=none,fill=color0] (axis cs:-1.68199969641788,0) rectangle (axis cs:-1.59989333818958,6);
\draw[draw=none,fill=color0] (axis cs:-1.59989333818958,0) rectangle (axis cs:-1.51778697996128,15);
\draw[draw=none,fill=color0] (axis cs:-1.51778697996128,0) rectangle (axis cs:-1.43568062173298,20);
\draw[draw=none,fill=color0] (axis cs:-1.43568062173298,0) rectangle (axis cs:-1.35357426350468,28);
\draw[draw=none,fill=color0] (axis cs:-1.35357426350468,0) rectangle (axis cs:-1.27146790527638,55);
\draw[draw=none,fill=color0] (axis cs:-1.27146790527638,0) rectangle (axis cs:-1.18936154704808,52);
\draw[draw=none,fill=color0] (axis cs:-1.18936154704808,0) rectangle (axis cs:-1.10725518881979,90);
\draw[draw=none,fill=color0] (axis cs:-1.10725518881979,0) rectangle (axis cs:-1.02514883059149,144);
\draw[draw=none,fill=color0] (axis cs:-1.02514883059149,0) rectangle (axis cs:-0.943042472363187,179);
\draw[draw=none,fill=color0] (axis cs:-0.943042472363187,0) rectangle (axis cs:-0.860936114134888,250);
\draw[draw=none,fill=color0] (axis cs:-0.860936114134888,0) rectangle (axis cs:-0.778829755906589,358);
\draw[draw=none,fill=color0] (axis cs:-0.778829755906589,0) rectangle (axis cs:-0.69672339767829,422);
\draw[draw=none,fill=color0] (axis cs:-0.69672339767829,0) rectangle (axis cs:-0.614617039449991,543);
\draw[draw=none,fill=color0] (axis cs:-0.614617039449991,0) rectangle (axis cs:-0.532510681221693,638);
\draw[draw=none,fill=color0] (axis cs:-0.532510681221693,0) rectangle (axis cs:-0.450404322993394,823);
\draw[draw=none,fill=color0] (axis cs:-0.450404322993394,0) rectangle (axis cs:-0.368297964765095,954);
\draw[draw=none,fill=color0] (axis cs:-0.368297964765095,0) rectangle (axis cs:-0.286191606536796,1080);
\draw[draw=none,fill=color0] (axis cs:-0.286191606536796,0) rectangle (axis cs:-0.204085248308497,1148);
\draw[draw=none,fill=color0] (axis cs:-0.204085248308497,0) rectangle (axis cs:-0.121978890080198,1180);
\draw[draw=none,fill=color0] (axis cs:-0.121978890080198,0) rectangle (axis cs:-0.0398725318518989,1286);
\draw[draw=none,fill=color0] (axis cs:-0.0398725318518989,0) rectangle (axis cs:0.0422338263764002,1321);
\draw[draw=none,fill=color0] (axis cs:0.0422338263764002,0) rectangle (axis cs:0.124340184604699,1319);
\draw[draw=none,fill=color0] (axis cs:0.124340184604699,0) rectangle (axis cs:0.206446542832998,1310);
\draw[draw=none,fill=color0] (axis cs:0.206446542832998,0) rectangle (axis cs:0.288552901061297,1137);
\draw[draw=none,fill=color0] (axis cs:0.288552901061297,0) rectangle (axis cs:0.370659259289595,1013);
\draw[draw=none,fill=color0] (axis cs:0.370659259289595,0) rectangle (axis cs:0.452765617517894,1008);
\draw[draw=none,fill=color0] (axis cs:0.452765617517894,0) rectangle (axis cs:0.534871975746193,803);
\draw[draw=none,fill=color0] (axis cs:0.534871975746193,0) rectangle (axis cs:0.616978333974493,677);
\draw[draw=none,fill=color0] (axis cs:0.616978333974493,0) rectangle (axis cs:0.699084692202792,546);
\draw[draw=none,fill=color0] (axis cs:0.699084692202792,0) rectangle (axis cs:0.78119105043109,460);
\draw[draw=none,fill=color0] (axis cs:0.78119105043109,0) rectangle (axis cs:0.863297408659389,327);
\draw[draw=none,fill=color0] (axis cs:0.863297408659389,0) rectangle (axis cs:0.945403766887688,225);
\draw[draw=none,fill=color0] (axis cs:0.945403766887688,0) rectangle (axis cs:1.02751012511599,185);
\draw[draw=none,fill=color0] (axis cs:1.02751012511599,0) rectangle (axis cs:1.10961648334429,126);
\draw[draw=none,fill=color0] (axis cs:1.10961648334429,0) rectangle (axis cs:1.19172284157258,98);
\draw[draw=none,fill=color0] (axis cs:1.19172284157258,0) rectangle (axis cs:1.27382919980088,63);
\draw[draw=none,fill=color0] (axis cs:1.27382919980088,0) rectangle (axis cs:1.35593555802918,38);
\draw[draw=none,fill=color0] (axis cs:1.35593555802918,0) rectangle (axis cs:1.43804191625748,29);
\draw[draw=none,fill=color0] (axis cs:1.43804191625748,0) rectangle (axis cs:1.52014827448578,16);
\draw[draw=none,fill=color0] (axis cs:1.52014827448578,0) rectangle (axis cs:1.60225463271408,7);
\draw[draw=none,fill=color0] (axis cs:1.60225463271408,0) rectangle (axis cs:1.68436099094238,3);
\draw[draw=none,fill=color0] (axis cs:1.68436099094238,0) rectangle (axis cs:1.76646734917068,2);
\draw[draw=none,fill=color0] (axis cs:1.76646734917068,0) rectangle (axis cs:1.84857370739898,4);
\draw[draw=none,fill=color0] (axis cs:1.84857370739898,0) rectangle (axis cs:1.93068006562728,3);
\draw[draw=none,fill=color0] (axis cs:1.93068006562728,0) rectangle (axis cs:2.01278642385557,2);
\end{axis}

\end{tikzpicture}
    \label{fig:distrib-1}
    }
\caption{Distribution of voters for one-dimensional Euclidean preferences}
\end{figure*}
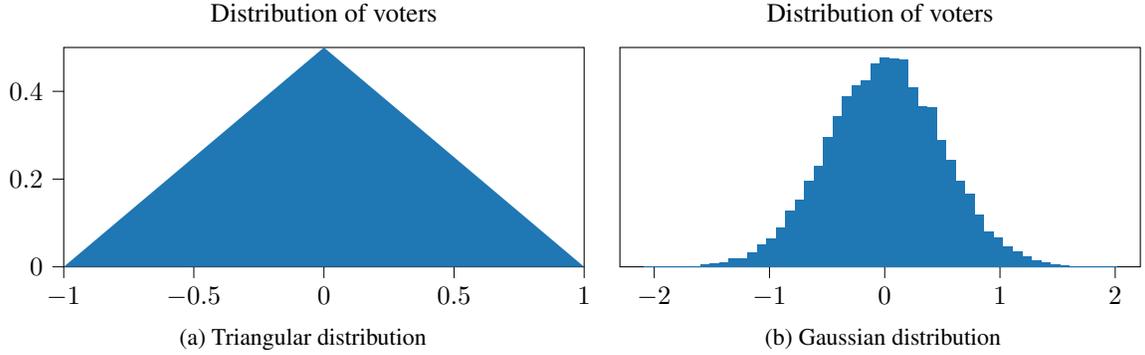

\begin{figure*}[!t]
    \centering
    \subfloat[Triangular distribution (Exact solution) \label{fig:results-0}]{
    \scalebox{0.95}{
\begin{tikzpicture}

\definecolor{color0}{rgb}{0.12156862745098,0.466666666666667,0.705882352941177}
\definecolor{color1}{rgb}{1,0.498039215686275,0.0549019607843137}
\definecolor{color2}{rgb}{0.172549019607843,0.627450980392157,0.172549019607843}
\definecolor{color3}{rgb}{0.83921568627451,0.152941176470588,0.156862745098039}

\begin{axis}[
legend cell align={left},
legend style={
  fill opacity=0.8,
  draw opacity=1,
  text opacity=1,
  at={(0.03,0.97)},
  anchor=north west,
  draw=white!80!black
},
tick align=outside,
tick pos=left,
title={Position of the second finalist for different values of d},
x grid style={white!69.0196078431373!black},
xlabel={$\alpha$},
xmin=0, xmax=1,
xtick style={color=black},
xtick={0,0.2,0.4,0.6,0.8,1},
xticklabels={0.0,0.2,0.4,0.6,0.8,1.0},
y grid style={white!69.0196078431373!black},
ylabel={distance to center},
ymin=0, ymax=1,
ytick style={color=black},
ytick={0,0.2,0.4,0.6,0.8,1},
yticklabels={0.0,0.2,0.4,0.6,0.8,1.0},
width=8.5cm,
height=4.5cm,
]
\addplot [very thick, color0]
table {%
0.01 0.00452261306532663
0.02 0.00909090909090909
0.03 0.0137055837563452
0.04 0.0183673469387755
0.05 0.0230769230769231
0.06 0.0278350515463918
0.07 0.0326424870466321
0.08 0.0375
0.09 0.0424083769633508
0.1 0.0473684210526316
0.11 0.0523809523809524
0.12 0.0574468085106383
0.13 0.0625668449197861
0.14 0.067741935483871
0.15 0.072972972972973
0.16 0.0782608695652174
0.17 0.0836065573770492
0.18 0.089010989010989
0.19 0.094475138121547
0.2 0.1
0.21 0.147619047619048
0.22 0.190909090909091
0.23 0.2
0.24 0.2
0.25 0.2
0.26 0.2
0.27 0.2
0.28 0.2
0.29 0.2
0.3 0.2
0.31 0.2
0.32 0.2
0.33 0.2
0.34 0.2
0.35 0.2
0.36 0.2
0.37 0.2
0.38 0.2
0.39 0.2
0.4 0.2
0.41 0.2
0.42 0.2
0.43 0.2
0.44 0.2
0.45 0.2
0.46 0.2
0.47 0.2
0.48 0.2
0.49 0.2
0.5 0.2
0.51 0.2
0.52 0.2
0.53 0.2
0.54 0.2
0.55 0.2
0.56 0.2
0.57 0.2
0.58 0.2
0.59 0.2
0.6 0.2
0.61 0.2
0.62 0.2
0.63 0.2
0.64 0.2
0.65 0.2
0.66 0.2
0.67 0.2
0.68 0.2
0.69 0.2
0.7 0.2
0.71 0.2
0.72 0.2
0.73 0.2
0.74 0.2
0.75 0.2
0.76 0.2
0.77 0.2
0.78 0.2
0.79 0.2
0.8 0.2
0.81 0.2
0.82 0.2
0.83 0.2
0.84 0.2
0.85 0.2
0.86 0.2
0.87 0.2
0.88 0.2
0.89 0.2
0.9 0.2
0.91 0.2
0.92 0.2
0.93 0.2
0.94 0.2
0.95 0.2
0.96 0.2
0.97 0.2
0.98 0.2
0.99 0.2
1 0.2
};
\addlegendentry{\footnotesize{d = 0.1}}
\addplot [very thick, color1]
table {%
0.01 0.00376884422110553
0.02 0.00757575757575758
0.03 0.0114213197969543
0.04 0.0153061224489796
0.05 0.0192307692307692
0.06 0.0231958762886598
0.07 0.0272020725388601
0.08 0.03125
0.09 0.0353403141361257
0.1 0.0394736842105263
0.11 0.0436507936507937
0.12 0.0478723404255319
0.13 0.0521390374331551
0.14 0.0564516129032258
0.15 0.0608108108108108
0.16 0.0652173913043478
0.17 0.069672131147541
0.18 0.0741758241758242
0.19 0.0787292817679558
0.2 0.0833333333333333
0.21 0.0879888268156425
0.22 0.0926966292134832
0.23 0.0974576271186441
0.24 0.102272727272727
0.25 0.107142857142857
0.26 0.112068965517241
0.27 0.117052023121387
0.28 0.122093023255814
0.29 0.12719298245614
0.3 0.132352941176471
0.31 0.137573964497041
0.32 0.142857142857143
0.33 0.148203592814371
0.34 0.153614457831325
0.35 0.159090909090909
0.36 0.164634146341463
0.37 0.170245398773006
0.38 0.175925925925926
0.39 0.18167701863354
0.4 0.1875
0.41 0.193396226415094
0.42 0.199367088607595
0.43 0.205414012738854
0.44 0.211538461538462
0.45 0.217741935483871
0.46 0.224025974025974
0.47 0.230392156862745
0.48 0.236842105263158
0.49 0.243377483443709
0.5 0.25
0.51 0.269607843137255
0.52 0.288461538461539
0.53 0.306603773584906
0.54 0.324074074074074
0.55 0.340909090909091
0.56 0.357142857142857
0.57 0.37280701754386
0.58 0.387931034482759
0.59 0.402542372881356
0.6 0.416666666666667
0.61 0.430327868852459
0.62 0.443548387096774
0.63 0.456349206349206
0.64 0.46875
0.65 0.480769230769231
0.66 0.492424242424242
0.67 0.5
0.68 0.5
0.69 0.5
0.7 0.5
0.71 0.5
0.72 0.5
0.73 0.5
0.74 0.5
0.75 0.5
0.76 0.5
0.77 0.5
0.78 0.5
0.79 0.5
0.8 0.5
0.81 0.5
0.82 0.5
0.83 0.5
0.84 0.5
0.85 0.5
0.86 0.5
0.87 0.5
0.88 0.5
0.89 0.5
0.9 0.5
0.91 0.5
0.92 0.5
0.93 0.5
0.94 0.5
0.95 0.5
0.96 0.5
0.97 0.5
0.98 0.5
0.99 0.5
1 0.5
};
\addlegendentry{\footnotesize{d = 0.25}}
\addplot [very thick, color2]
table {%
0.01 0.00336683417085427
0.02 0.00676767676767677
0.03 0.0102030456852792
0.04 0.0136734693877551
0.05 0.0171794871794872
0.06 0.0207216494845361
0.07 0.024300518134715
0.08 0.0279166666666667
0.09 0.0315706806282723
0.1 0.0352631578947368
0.11 0.038994708994709
0.12 0.0427659574468085
0.13 0.0465775401069519
0.14 0.0504301075268817
0.15 0.0543243243243243
0.16 0.0582608695652174
0.17 0.0622404371584699
0.18 0.0662637362637363
0.19 0.0703314917127072
0.2 0.0744444444444444
0.21 0.0786033519553073
0.22 0.0828089887640449
0.23 0.0870621468926554
0.24 0.0913636363636364
0.25 0.0957142857142857
0.26 0.100114942528736
0.27 0.104566473988439
0.28 0.10906976744186
0.29 0.113625730994152
0.3 0.118235294117647
0.31 0.122899408284024
0.32 0.127619047619048
0.33 0.132395209580838
0.34 0.137228915662651
0.35 0.142121212121212
0.36 0.147073170731707
0.37 0.152085889570552
0.38 0.15716049382716
0.39 0.162298136645963
0.4 0.1675
0.41 0.172767295597484
0.42 0.178101265822785
0.43 0.183503184713376
0.44 0.188974358974359
0.45 0.194516129032258
0.46 0.20012987012987
0.47 0.205816993464052
0.48 0.211578947368421
0.49 0.217417218543046
0.5 0.223333333333333
0.51 0.229328859060403
0.52 0.235405405405405
0.53 0.24156462585034
0.54 0.247808219178082
0.55 0.254137931034483
0.56 0.260555555555556
0.57 0.267062937062937
0.58 0.273661971830986
0.59 0.280354609929078
0.6 0.287142857142857
0.61 0.294028776978417
0.62 0.301014492753623
0.63 0.308102189781022
0.64 0.315294117647059
0.65 0.322592592592593
0.66 0.33
0.67 0.344925373134328
0.68 0.359411764705882
0.69 0.373478260869565
0.7 0.387142857142857
0.71 0.400422535211268
0.72 0.413333333333333
0.73 0.425890410958904
0.74 0.438108108108108
0.75 0.45
0.76 0.461578947368421
0.77 0.472857142857143
0.78 0.483846153846154
0.79 0.494556962025317
0.8 0.505
0.81 0.515185185185185
0.82 0.525121951219512
0.83 0.534819277108434
0.84 0.544285714285714
0.85 0.553529411764706
0.86 0.562558139534884
0.87 0.571379310344828
0.88 0.58
0.89 0.588426966292135
0.9 0.596666666666667
0.91 0.604725274725275
0.92 0.612608695652174
0.93 0.620322580645161
0.94 0.627872340425532
0.95 0.635263157894737
0.96 0.6425
0.97 0.649587628865979
0.98 0.656530612244898
0.99 0.66
1 0.66
};
\addlegendentry{\footnotesize{d = 0.33}}
\addplot [very thick, color3]
table {%
0.01 0.00251256281407035
0.02 0.00505050505050505
0.03 0.00761421319796954
0.04 0.0102040816326531
0.05 0.0128205128205128
0.06 0.0154639175257732
0.07 0.0181347150259067
0.08 0.0208333333333333
0.09 0.0235602094240838
0.1 0.0263157894736842
0.11 0.0291005291005291
0.12 0.0319148936170213
0.13 0.0347593582887701
0.14 0.0376344086021505
0.15 0.0405405405405405
0.16 0.0434782608695652
0.17 0.046448087431694
0.18 0.0494505494505495
0.19 0.0524861878453039
0.2 0.0555555555555556
0.21 0.058659217877095
0.22 0.0617977528089888
0.23 0.0649717514124294
0.24 0.0681818181818182
0.25 0.0714285714285714
0.26 0.0747126436781609
0.27 0.0780346820809249
0.28 0.0813953488372093
0.29 0.0847953216374269
0.3 0.0882352941176471
0.31 0.0917159763313609
0.32 0.0952380952380952
0.33 0.0988023952095809
0.34 0.102409638554217
0.35 0.106060606060606
0.36 0.109756097560976
0.37 0.113496932515337
0.38 0.117283950617284
0.39 0.12111801242236
0.4 0.125
0.41 0.128930817610063
0.42 0.132911392405063
0.43 0.136942675159236
0.44 0.141025641025641
0.45 0.145161290322581
0.46 0.149350649350649
0.47 0.15359477124183
0.48 0.157894736842105
0.49 0.162251655629139
0.5 0.166666666666667
0.51 0.171140939597315
0.52 0.175675675675676
0.53 0.180272108843537
0.54 0.184931506849315
0.55 0.189655172413793
0.56 0.194444444444444
0.57 0.199300699300699
0.58 0.204225352112676
0.59 0.209219858156028
0.6 0.214285714285714
0.61 0.219424460431655
0.62 0.22463768115942
0.63 0.22992700729927
0.64 0.235294117647059
0.65 0.240740740740741
0.66 0.246268656716418
0.67 0.25187969924812
0.68 0.257575757575758
0.69 0.263358778625954
0.7 0.269230769230769
0.71 0.275193798449612
0.72 0.28125
0.73 0.28740157480315
0.74 0.293650793650794
0.75 0.3
0.76 0.306451612903226
0.77 0.313008130081301
0.78 0.319672131147541
0.79 0.326446280991736
0.8 0.333333333333333
0.81 0.340336134453782
0.82 0.347457627118644
0.83 0.354700854700855
0.84 0.362068965517241
0.85 0.369565217391304
0.86 0.37719298245614
0.87 0.384955752212389
0.88 0.392857142857143
0.89 0.400900900900901
0.9 0.409090909090909
0.91 0.417431192660551
0.92 0.425925925925926
0.93 0.434579439252337
0.94 0.443396226415094
0.95 0.452380952380952
0.96 0.461538461538462
0.97 0.470873786407767
0.98 0.480392156862745
0.99 0.49009900990099
1 0.5
};
\addlegendentry{\footnotesize{d = 0.5}}
\addplot [very thick, black, dashed]
table {%
0.5 0
0.5 1
};
\addlegendentry{\footnotesize{S-PAV}}
\end{axis}

\end{tikzpicture}}
    }
    \subfloat[Gaussian distribution (Empirical solution) \label{fig:results-1}]{
     \scalebox{0.95}{\input{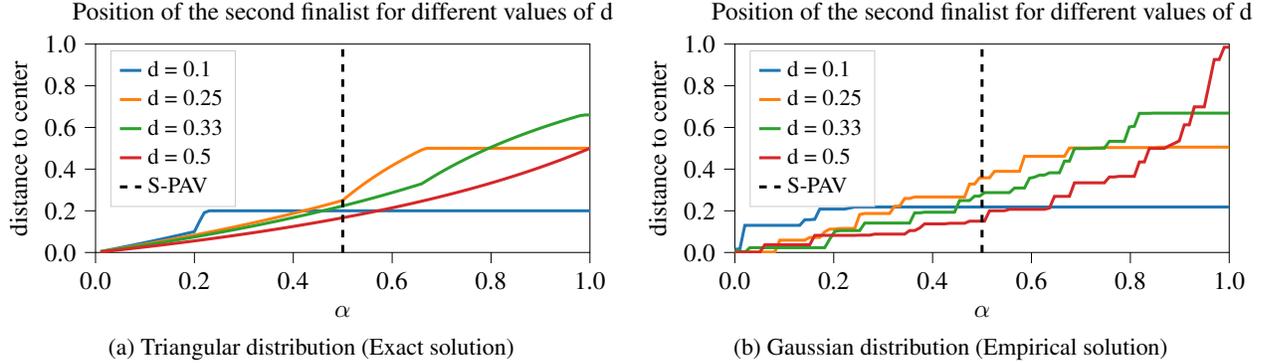}}
    }
\caption{Position of the second finalist for $\alpha$-seqAV rules for various values of the approval radius $d$ and with $\alpha \in [0,1]$. The position corresponds to the distance to the center of the distribution (which is also the position of the first finalist)}
    \label{fig:results-sp}
\end{figure*}

\subsection{Triangular distribution}
We first assume that the distribution of voters follows the following triangular density function $f$: all voters and candidates are located on $[-1,1]$, and for all
$x \in [-1,1]$, $f(x) = \frac{1-|x|}{2}$. The distribution is showed in Figure \ref{fig:distrib-0}. 

Let us consider  $\alpha$-seqAV rules, for $\alpha \in [0,1]$: for all pairs of finalists, one finalist $x_1$ is an approval winner and the other finalist $x_2$ maximizes $S_V(x_2) - \alpha S_V(x_1x_2)$. Recall that $\alpha = 0$ leads to MAV and $\alpha = 1$ to S-CCAV. The first finalist will always be the closest to the middle point of the interval (here $0$). It can be shown that the position in which the second finalist
will get the maximum score for specific $\alpha$ and $d$ is
\[   
|x_2^*| = 
     \begin{cases}
       \frac{\alpha(1-d)}{2-\alpha} & \text{if }  \alpha \le 2d \\
      1+d-\frac{2d}{\alpha} &\text{if } 2d \le \alpha \le \frac{2d}{1-d}\\
      2d &\text{if } \alpha > \frac{2d}{1-d}\\
     \end{cases}
\]
\begin{proof}
We know that the distribution of voters follow the function $f(x) = \frac{1-|x|}{2}$ and the first finalist $w$ is at position $\phi(w) = 0$. Let $c$ be a candidate at position $\phi(c) = x$. 
\par 
First, consider the case $x \le d$, then $-d \le x-d \le 0$ and the score of $c$ is given by
\begin{align*}
S(c) =& S_V(c) - \alpha S_V(cw) \\
=&\int_{x-d}^{x+d} f(t)dt - \alpha \int_{x-d}^d f(t)dt \\
=& \frac{1}{2}((1-\alpha)\int_{x-d}^{0} (1+t)dt + \int_{0}^{x+d}(1-t)dt + \int_{0}^{d}(1-t)dt) \\
=& \frac{1}{2}((1-\alpha)(0-(x-d +\frac{(x-d)^2}{2})) + (x+d - \frac{(x+d)^2}{2})-0)+ C\\
=& \frac{1}{2}((1-\alpha)(-x - \frac{x^2 - 2xd + d^2}{2} + d) + x + d - \frac{x^2 + 2dx + d^2}{2}) + C\\
=& \frac{1}{2}(x^2( \frac{\alpha-1}{2} - \frac{1}{2}) + x((1-\alpha)(-1+d) + 1 - d)) + C\\
=& \frac{1}{2}(x^2\frac{\alpha-2}{2} + x \alpha( 1-d)) + C
\end{align*}
where $C$ is a constant independent of $x$. The maximum of this function is reached when $x^* = \frac{\alpha(1-d)}{2-\alpha}$, which is the optimal position. Note that $x^* \le d$ if and only if $\alpha(1-d) \le (2-\alpha)d$, and therefore $\alpha \le 2d$. This gives us the first formula.
\par 
Let's now assume that $2d \ge x \ge d$. The score of $c$ is now given by
\begin{align*}
S(c) =& S_V(c) - \alpha S_V(cw) \\
&= \int_{x-d}^{x+d} f(t)dt - \alpha \int_{x-d}^d f(t)dt \\
&= (1-\alpha)\int_{x-d}^{d} (1-t)dt + \int_{d}^{x+d} (1-t)dt \\
&= \frac{1}{2}((1-\alpha)((d -\frac{d^2}{2})- (x-d - \frac{(x-d)^2}{2}))+ (x+d - \frac{(x+d)^2}{2} - (d - \frac{d^2}{2}))) \\
&= \frac{1}{2}(x^2(\frac{(1-\alpha)}{2} - \frac{1}{2}) + x( (1-\alpha)(-1  - d) + 1 - d)) + C \\
&= \frac{1}{2}(x^2\frac{-\alpha}{2} + x(\alpha(1+d) - 2d)) + C
\end{align*}
where $C$ is a constant independent of $x$. The maximum of this function is reached when $x^* = \frac{\alpha(1+d) - 2d}{\alpha} = 1 + d - \frac{2d}{\alpha}$. Note that $d \le x^* \le 2d$ if and only if $0 \le 1 - \frac{2d}{\alpha} \le d$ so $\frac{2d}{1-d} \ge \alpha \ge 2d$. This gives the second formula.
\par 
Finally, if $x^* \ge 2d$ then $S(x) = \int_{x-d}^{x+d}f(t)dt$ and since $f(t)$ is decreasing, this is maximum when $x^* = 2d$. Therefore, for $\alpha > \frac{2d}{1-d}$ we have $x^* = 2d$. This gives the third formula and concludes the proof.
\end{proof}

This optimal distance to the first finalist is depicted on Figure \ref{fig:results-0}. This figure clearly shows that for rules close to MAV, the second finalist is quite centrist, and the closer we are to S-CCAV, the more extreme is the second finalist. This being said, $\alpha$-seqAV rules starts to be equivalent to S-CCAV before reaching $\alpha = 1$ when $d < \frac{1}{3}$.

\subsection{Gaussian distribution}
Now, we assume that we have a Gaussian distribution of voters, centered at $0$ and with standard deviation $1/2$, as shown in Figure \ref{fig:distrib-1}. We cannot compute the exact optimum for each $d$ and $\alpha$ so we used simulations on
synthetic data instead. We sampled $20,000$ voters and $1,000$ candidates and as before, every voter approves candidates that are at distance $\le d$. For each $d$ and $\alpha$ we compute the two finalists and observed their positions on the line. The first selected finalist is always the closest to the center, and Figure \ref{fig:results-1} shows the distance between the two finalists for $\alpha$-seqAV rules with $\alpha \in [0,1]$ and various $d$. We observe that Figure \ref{fig:results-1} is very similar to Figure \ref{fig:results-0}.

\section{Experiments}
\label{experiments}

Finally, we want to compare the different rules on real data.
We used approval ballot datasets from different sources\footnote{The last two datasets are available on \url{www.preflib.org}}:
\begin{itemize}
    \item Datasets from the poll "Voter Autrement" conducted in several cities
    during the 2017 French presidential election \citep{bouveret:halshs-02379941}. Each dataset contains approval ballots of around $1000$ voters for the $11$ running candidates.
   
    \item A dataset from the 2002 French presidential election, which had $16$ candidates. 
    \item Two datasets 
    of the San Sebastian Poster Competition held during The Summer School on Computational Social Choice ($17$ candidates,  around $60$ voters per dataset).
\end{itemize} 
Table \ref{tab:dataset} summarizes our different datasets.
\begin{table}[!t]
    \centering
    \begin{tabular}{|c|c|c|} \hline 
        Name & Candidates & Voters  \\\hline
        2017-Strasbourg & $11$ & $1055$ \\\hline
        2017-HSC & $11$ & $701$\\\hline 
        2017-Grenoble & $11$ & $1048$\\\hline
        2017-Crolles-1 & $11$ & $1291$\\\hline
        2017-Crolles-2 & $11$ & $1269$ \\\hline
        2002-Presidential & $16$ & $2597$ \\\hline
        Best-Poster-A & $17$ & $65$ \\\hline
        Best-Poster-B & $17$ & $59$ \\\hline
    \end{tabular}
    \caption{The datasets we used. For 2017-Crolles, half of the voters responded thinking there will be no runoff (2017-Crolles-1), and the other half thinking there will be a runoff (2017-Crolles-2).}
    \label{tab:dataset}
\end{table}
\par 
The first thing we did was to debiaised the results of the 2017 datasets. Indeed, the pollsters asked respondents who they actually voted in the election, and the distribution of votes is very different between the poll respondents and the global results. For instance, in the Strasbourg dataset, $35\%$ of the voters indicated that they voted JLM in the election, against $20\%$ nationwide. Similarly, $3\%$ of the voters in this poll indicated they voted MLP, against $21\%$ nationwide. 
\par 
After this debiaising step, we can look at the approval score of candidates in the election. For instance, is there a candidate with more than $50\%$ approvals? Table \ref{tab:approval_scores17} and \ref{tab:approval_scores02} summarize the approval scores of the main candidates in the different datasets of 2002 and 2017 elections. We can see that it is actually tight between the different candidates.
\begin{table}[!t]
    \centering
    \begin{tabular}{|l|c|c|c|c|c|}
    \hline
        & \meluche & \hamon & \macron & \fillon & \lepen  \\ \hline \hline
        Real &  $ 19.6\%$ & $6.4\%$& $24\%$ & $20\%$ & $21.3\%$\\ \hline \hline
        Strasbourg & $45.5\%$& $37.5\%$ & $46.3\%$& $26.7\%$& $25.8\%$ \\ \hline
        HSC & $ 41.2\%$ & $36.6\%$ & $45.6\%$ & $29.4\%$ & $28.7\% $ \\ \hline
        Grenoble  & $41.8\%$ & $ 43.2\%$ & $42.3\%$ & $33.2\%$& $29.5\%$  \\ \hline
        Crolles-1 & $40.5\%$& $36.3\%$& $54.7\%$& $34\%$& $30.3\%$  \\ \hline
        Crolles-2  & $43.6\%$& $40.7\%$& $55.3\%$& $38.7\%$& $31.4\%$ \\ \hline
    \end{tabular}
    \caption{Approval scores of candidates in the 2017 election datasets}
    \label{tab:approval_scores17}
\end{table}
\begin{table}[!t]
    \centering
    \begin{tabular}{|l|c|c|c|c|c|c|}
    \hline
        & \jospin & \green & \bayrou & \chirac & \lepen  \\ \hline \hline
        Real &  $ 16.2\%$ & $5.2\%$& $6.8\%$ & $19.9\%$ & $16.9\%$\\ \hline
        Approvals & $40.5\%$& $28.8\%$ & $33.4\%$& $36.4\%$& $14.6\%$ \\ \hline
    \end{tabular}
    \caption{Approval scores of candidates in the 2002 election dataset}
    \label{tab:approval_scores02}
\end{table}
\par 
We also played with the distribution of the approval ballots to understand the affinity network of the candidates. Two candidates are closer if voters that approve one of them tend to approve the second one. We compute the affinity between two candidates with the Jaccard index, i.e. we divide the number of voters who approved both candidates by the number of voters who approved at least one of them. Figure \ref{fig:network-grenoble} shows the affinity network for the 2017 presidential election in Grenoble. The size of a node is proportional to its approval score. We can clearly see an group with candidates from the left (in red), and another one with candidates from the right (in blue), EM being somewhat in the middle of the two groups. The affinity networks of all datasets can be found in the Appendix
\begin{figure}[!t]
    \centering
    \includegraphics[width=0.9\textwidth]{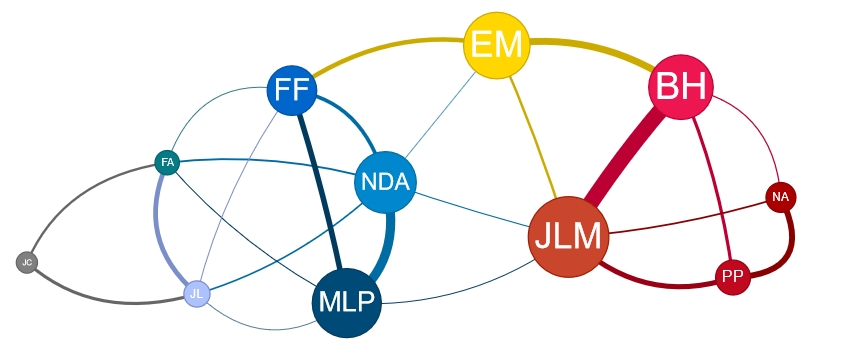}
    \caption{Affinity network for the dataset 2017-Grenoble with a threshold at 10\%, i.e. the Jacquard index must be $> 0.1$ for the edge to be visible}
    \label{fig:network-grenoble}
\end{figure}

\par 
We ran the different rules presented in Section \ref{sec:rules} on our datasets. Table \ref{tab:result_unbiaised} summarizes the finalists obtained for each dataset and each rule. For some datasets, the choice of the rule has a strong impact on the finalists, which suggests that the choice of the ABC rule should be made with care.
\begin{table}[!t]
    \centering
    \begin{tabular}{|c|c|c|c|c|c|c|c|c} \hline
         Rule &  SAV & CCAV & S-CCAV & PAV & S-PAV& S-Phr & SAV\\\hline
         2017-Strasbourg & \macron/\meluche   & \macron/\meluche &  \macron/ \meluche  &  \macron/\meluche& \macron/\meluche  &\macron/ \meluche &\hamon/ \meluche \\\hline
         2017-Grenoble & \hamon/\macron   & \hamon/\fillon & \hamon/\fillon  &  \macron/\meluche& \macron/\hamon  &\fillon/ \hamon &\macron/ \hamon \\\hline
         2017-HSC & \macron/\meluche   & \macron/\meluche &  \macron/ \meluche  &  \macron/\meluche& \macron/\meluche  &\macron/ \meluche &\macron/ \meluche \\\hline
         2017-Crolles-1 & \macron/\meluche   & \macron/\lepen &  \macron/ \lepen  &  \macron/\meluche& \macron/\meluche  &\macron/ \meluche &\macron/ \meluche \\\hline
         2017-Crolles-2 &\macron/\meluche   & \macron/\lepen &  \macron/ \lepen  &  \macron/\meluche& \macron/\meluche  &\macron/ \meluche &\macron/ \meluche \\\hline\hline
         2002-Presidential &\jospin/\chirac   & \jospin/\chirac & \jospin/\chirac   &  \jospin/\chirac & \jospin/\chirac  &\jospin/\chirac  &\jospin/\chirac  \\\hline\hline
         Best-Poster-A &$\#1$/$\#2$   & $\#1$/$\#6$ &$\#1$/$\#6$   &  $\#1$/$\#4$ & $\#1$/$\#4$ &$\#1$/$\#4$  &$\#1$/$\#2$ \\ \hline
         Best-Poster-B &$\#1$/$\#2$   & $\#1$/$\#2$ &$\#1$/$\#2$   & $\#1$/$\#2$ & $\#1$/$\#2$&$\#1$/$\#2$ &$\#1$/$\#2$
         \\\hline
    \end{tabular}
    \caption{Finalists with different approval rules on all our datasets}
    \label{tab:result_unbiaised}
\end{table}
\par
We can also look at the evolution of the finalists with $\alpha$-AV and $\alpha$-seqAV rules when $\alpha$ varies from $0$ to $1$. This gives us a spectrum of rules from MAV to (S-)CCAV and we can see how the results evolve between the extremes. For MAV, the finalists are the two candidates from the strongest group, and for CCAV, they are from two very different groups.
\begin{figure*}[!t]
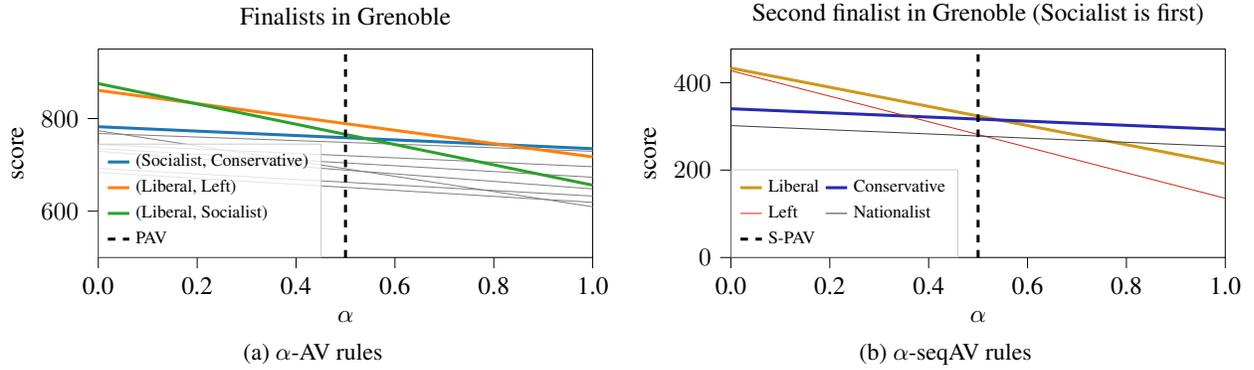

    \centering
    \subfloat[$\alpha$-AV rules\label{fig:spectrum2grenoble}]{
    \scalebox{0.95}{\input{tikz/spectrum_2_Grenoble_main}}
    }
    \subfloat[$\alpha$-seqAV rules\label{fig:spectrum1grenoble}]{
    \scalebox{0.95}{\input{tikz/spectrum_1_Grenoble_main}}
    }
    \caption{The two finalists in 2017-Grenoble Dataset for $\alpha$-AV and $\alpha$-seqAV rules with $\alpha \in [0,1]$. For every $\alpha$, the upper line indicates the finalists (resp. the second finalist) for the $\alpha$-AV (resp. $\alpha$-seqAV) rule.}
    \label{fig:spectrumGrenoble}
\end{figure*}
\par 
Figure \ref{fig:spectrum2grenoble} depicts for instance the evolution of the pair of finalists for $\alpha$-AV rules for the 2017-Grenoble Dataset. We can see that the pairs of finalists change twice and involve $4$ different candidates (Liberal, Left, Socialist and Conservative). Pairs that are never selected are represented by grey lines.

\par 
Figure \ref{fig:spectrum1grenoble} depicts the evolution of the $\alpha$-seqAV score of candidates for second finalist spot in the 2017-Grenoble dataset. The approval winner is the Socialist candidate; the second finalist is either the Liberal or the Conservative.

\par 
More results and figures for these datasets, and for other datasets, are in the Appendix.

\section{Conclusion}
\label{conclu}

Our main message is that {\em approval with runoff} is not {\em one} rule but a {\em family} of rules, parameterized by the ABC rule chosen for determining the finalists. Our axiomatic and experimental results in Sections \ref{normative}, \ref{sp} and  \ref{experiments} show that this choice does make a big difference. If such rules have to be used in political elections, the choice of the ABC rule will be crucial, and is far from easy, but our results already give some useful elements: on a political election on a single-peaked profile on a classical left-right spectrum, we conclude from our theoretical and experimental results that: multiwinner approval voter tends to select two centrist finalists (possibly clones of each other); Chamberlin-Courant tends to select a left-wing and a right-wing candidate; proportional approval voting is inbetween and tends to select a center-left and/or a center-right candidate; greedy versions of the latter two rules select the candidate with the largest number of approvals, and a left-wing or right-wing candidate. 

An important question is, will citizens understand and accept such rules especially in comparison with plurality with runoff and standard (single-winner) approval voting? Will there be a difference between citizens used to runoff voting in their country and those who are not? 

\bibliographystyle{unsrtnat}
\bibliography{biblio}
\newpage
\appendix
\section{Affinity networks}
\begin{figure}[!h]
    \centering
    \includegraphics[width=0.7\textwidth]{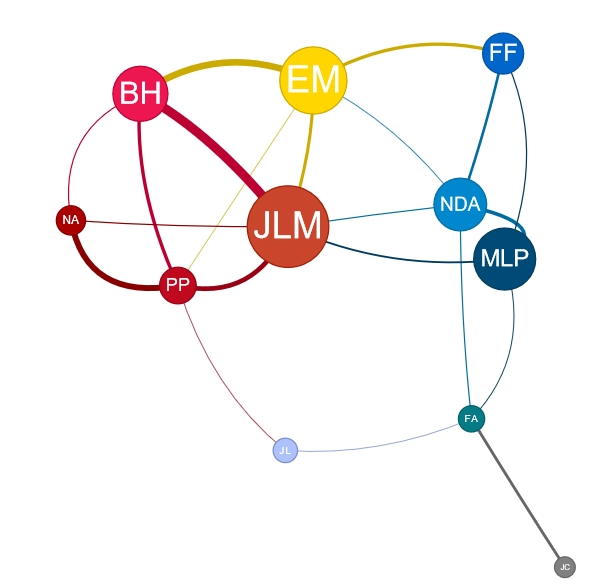}
    \caption{Affinity network for the dataset 2017-Strasbourg with a threshold at 10\%}
\end{figure}
\begin{figure}[!h]
    \centering
    \includegraphics[width=0.7\textwidth]{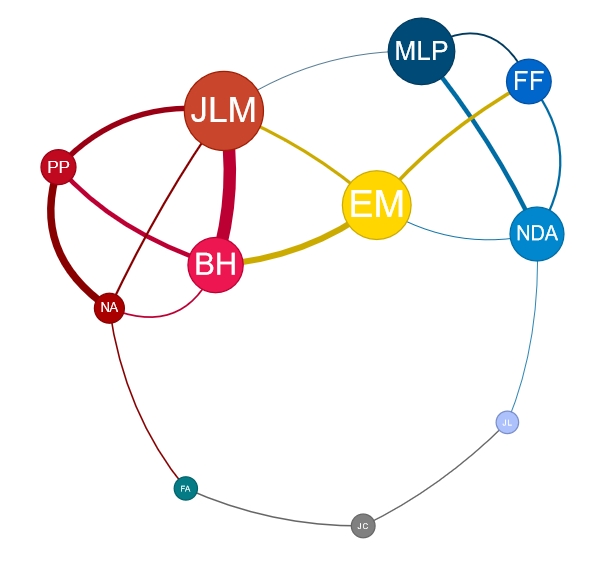}
    \caption{Affinity network for the dataset 2017-HSC with a threshold at 10\%}
\end{figure}
\begin{figure}[!h]
    \centering
    \includegraphics[width=0.7\textwidth]{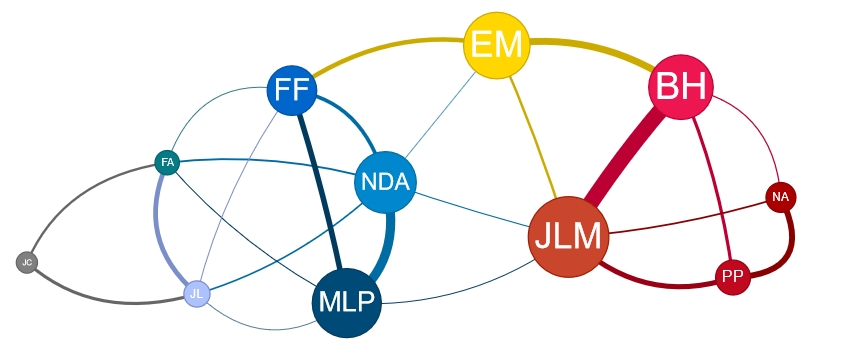}
    \caption{Affinity network for the dataset 2017-Grenoble with a threshold at 10\%}
\end{figure}
\begin{figure}[!h]
    \centering
    \includegraphics[width=0.7\textwidth]{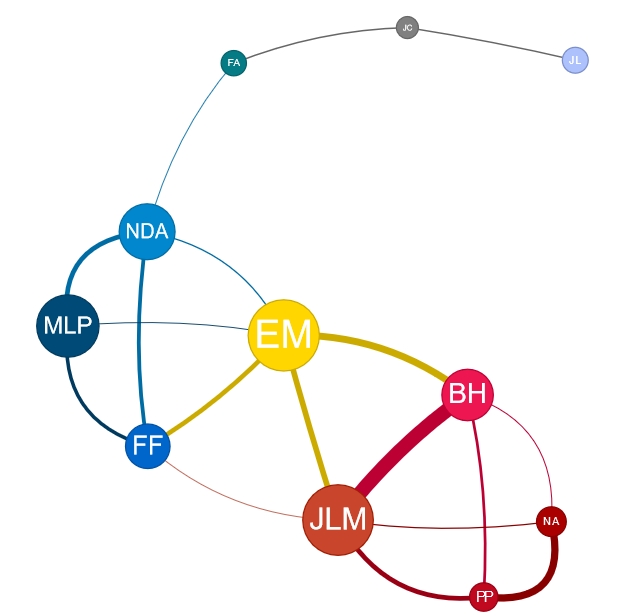}
    \caption{Affinity network for the dataset 2017-Crolles-1 with a threshold at 10\%}
\end{figure}
\begin{figure}[!h]
    \centering
    \includegraphics[width=0.7\textwidth]{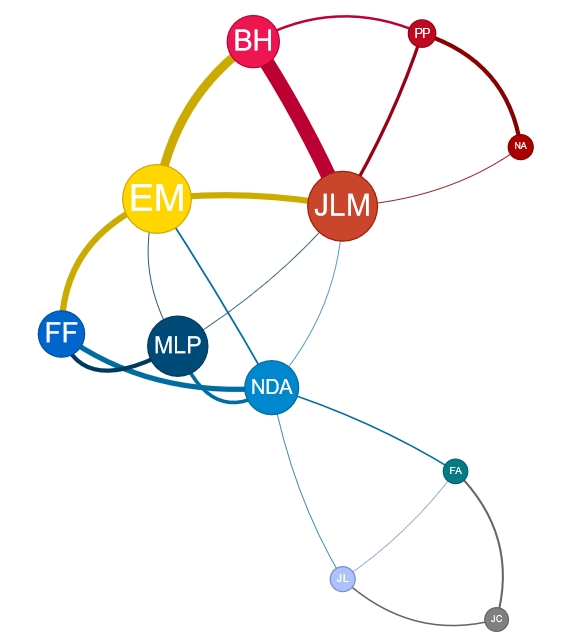}
    \caption{Affinity network for the dataset 2017-Crolles-2 with a threshold at 10\%}
\end{figure}
\begin{figure}[!h]
    \centering
    \includegraphics[width=0.7\textwidth]{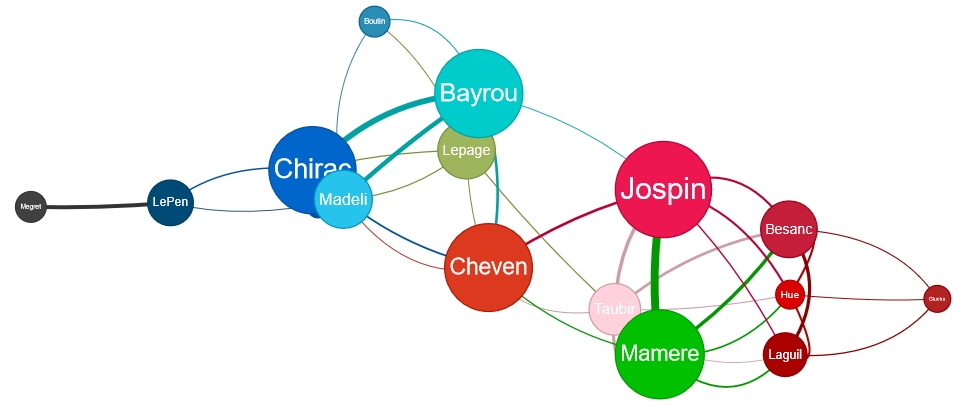}
    \caption{Affinity network for the dataset 2002-Presidential with a threshold at 12\%}
\end{figure}

\begin{figure}[!h]
    \centering
    \includegraphics[width=0.7\textwidth]{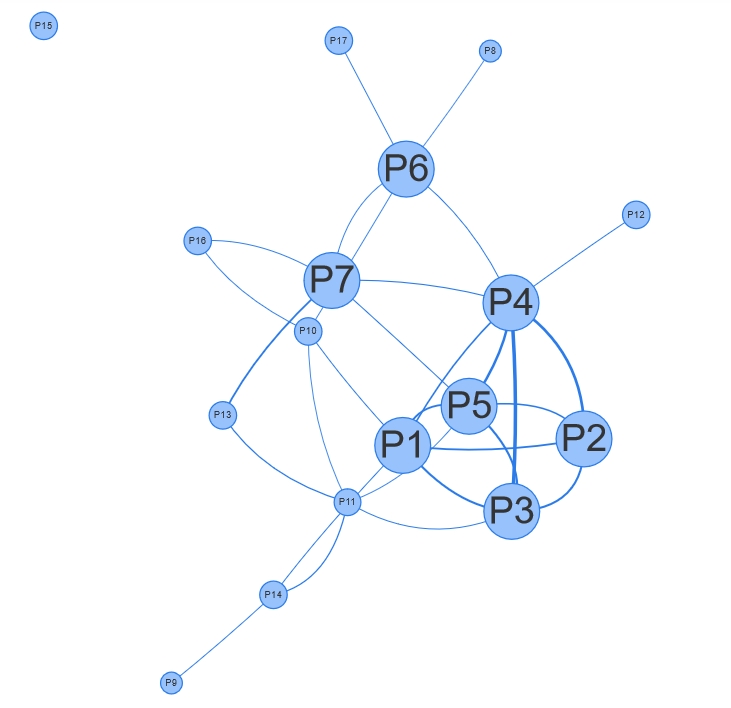}
    \caption{Affinity network for the dataset Poster-1 with a threshold at 30\%}
\end{figure}

\begin{figure}[!h]
    \centering
    \includegraphics[width=0.7\textwidth]{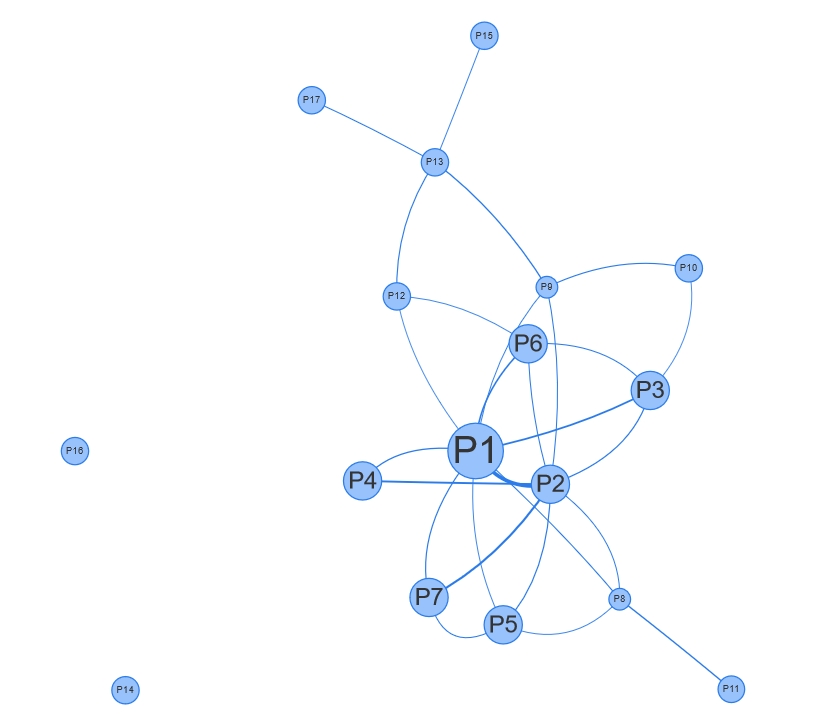}
    \caption{Affinity network for the dataset Poster-2 with a threshold at 30\%}
\end{figure}

\clearpage

\section{Results for all datasets}


\begin{figure*}[h]
    \centering
    \subfloat[$\alpha$-AV rules]{
    \input{tikz/spectrum_2_Grenoble}
    }
    \subfloat[$\alpha$-seqAV rules]{
    \input{tikz/spectrum_1_Grenoble}
    }
\caption{The two finalists in 2017-Grenoble Dataset for $\alpha$-AV and $\alpha$-seqAV rules with $\alpha \in [0,1]$. For every $\alpha$, the upper line indicates the finalists (resp. the second finalist) for the $\alpha$-AV (resp. $\alpha$-seqAV) rule.}
\label{fig:alpha_1}
\end{figure*}

\begin{figure*}[h]
    \centering
    \subfloat[$\alpha$-AV rules]{
    \input{tikz/spectrum_2_Crolles1}
    }
    \subfloat[$\alpha$-seqAV rules]{
    \input{tikz/spectrum_1_Crolles1}
    }
\caption{The two finalists in 2017-Crolles-1 Dataset for $\alpha$-AV and $\alpha$-seqAV rules with $\alpha \in [0,1]$. For every $\alpha$, the upper line indicates the finalists (resp. the second finalist) for the $\alpha$-AV (resp. $\alpha$-seqAV) rule.}
\end{figure*}

\begin{figure*}[p]
    \centering
    \subfloat[$\alpha$-AV rules]{
    \input{tikz/spectrum_2_Crolles2}
    }
    \subfloat[$\alpha$-seqAV rules]{
    \input{tikz/spectrum_1_Crolles2}
    }
\caption{The two finalists in 2017-Crolles-2 Dataset for $\alpha$-AV and $\alpha$-seqAV rules with $\alpha \in [0,1]$. For every $\alpha$, the upper line indicates the finalists (resp. the second finalist) for the $\alpha$-AV (resp. $\alpha$-seqAV) rule.}
\end{figure*}

\begin{figure*}[p]
    \centering
    \subfloat[$\alpha$-AV rules]{
    \input{tikz/spectrum_2_HSC}
    }
    \subfloat[$\alpha$-seqAV rules]{
    \input{tikz/spectrum_1_HSC}
    }
\caption{The two finalists in 2017-HSC Dataset for $\alpha$-AV and $\alpha$-seqAV rules with $\alpha \in [0,1]$. For every $\alpha$, the upper line indicates the finalists (resp. the second finalist) for the $\alpha$-AV (resp. $\alpha$-seqAV) rule.}
\end{figure*}

\begin{figure*}[p]
    \centering
    \subfloat[$\alpha$-AV rules]{
    \input{tikz/spectrum_2_Strasbourg}
    }
    \subfloat[$\alpha$-seqAV rules]{
    \input{tikz/spectrum_1_Strasbourg}
    }
\caption{The two finalists in 2017-Strasbourg Dataset for $\alpha$-AV and $\alpha$-seqAV rules with $\alpha \in [0,1]$. For every $\alpha$, the upper line indicates the finalists (resp. the second finalist) for the $\alpha$-AV (resp. $\alpha$-seqAV) rule.}
\end{figure*}

\begin{figure*}[p]
    \centering
    \subfloat[$\alpha$-AV rules]{
    \input{tikz/spectrum_AV2002_2}
    }
    \subfloat[$\alpha$-seqAV rules]{
    \input{tikz/spectrum_AV2002}
    }
\caption{The two finalists in the 2002-Presidential dataset for $\alpha$-AV and $\alpha$-seqAV rules with $\alpha \in [0,1]$. For every $\alpha$, the upper line indicates the finalists (resp. the second finalist) for the $\alpha$-AV (resp. $\alpha$-seqAV) rule.}
\end{figure*}

\begin{figure*}[p]
    \centering
    \subfloat[$\alpha$-AV rules]{
    \input{tikz/spectrum_poster_2_0}
    }
    \subfloat[$\alpha$-seqAV rules]{
    \input{tikz/spectrum_poster_0}
    }
\caption{The two finalists in the Best-Poster-A dataset for $\alpha$-AV and $\alpha$-seqAV rules with $\alpha \in [0,1]$. For every $\alpha$, the upper line indicates the finalists (resp. the second finalist) for the $\alpha$-AV (resp. $\alpha$-seqAV) rule.}
\end{figure*}

\begin{figure*}[p]
    \centering
    \subfloat[$\alpha$-AV rules]{
    \input{tikz/spectrum_poster_2_1}
    }
    \subfloat[$\alpha$-seqAV rules]{
    \input{tikz/spectrum_poster_1}
    }
\caption{The two finalists in the Best-Poster-B dataset for $\alpha$-AV and $\alpha$-seqAV rules with $\alpha \in [0,1]$. For every $\alpha$, the upper line indicates the finalists (resp. the second finalist) for the $\alpha$-AV (resp. $\alpha$-seqAV) rule.}
\label{fig:alpha_n}
\end{figure*}

\end{document}